\documentclass[lettersize,journal]{IEEEtran}
\usepackage{amsmath,amsfonts, amsthm}
\usepackage{algorithmic}
\usepackage{algorithm}
\usepackage{array}
\usepackage{amssymb} %
\usepackage{enumitem}
\usepackage{textcomp}
\usepackage{stmaryrd}
\usepackage{lipsum}
\usepackage{stfloats}
\usepackage{verbatim}
\usepackage{graphicx}
\usepackage{cite}
\usepackage{hyperref} 
\usepackage{etoolbox}
\usepackage{booktabs}
\usepackage{mathtools}
\usepackage{makecell} 
\usepackage{multirow}
\hypersetup{colorlinks=true,linkcolor=blue,citecolor=blue}
\usepackage{xcolor}

\newcommand{\myref}[1]{\textcolor{blue}{(\ref{#1})}}

\newtheoremstyle{mystyle}{}{}
{}{}{}
{.}{3pt}{{\it{\thmname{#1} \thmnumber{#2}}\thmnote{(#3)}}}
\theoremstyle{mystyle}
\newtheorem{proposition}{Proposition}

\usepackage{subcaption}
\begin{document}

\title{A Privacy-Preserving Localization Scheme with Node Selection in Mobile Networks}

\author{Liangbo Xie, \textit{Member, IEEE}, Mude Cai, Xiaolong Yang, Mu Zhou, \textit{Senior Member, IEEE} Jiacheng Wang, and Dusit Niyato, \textit{Fellow, IEEE}  
\thanks{This work is supported in part by the National Natural Science Foundation of China (62101085, 62571074), the Chongqing Natural Science Foundation Project (CSTB2023NSCQ-MSX0249, CSTB2025NSCQ-LZX0142, CSTB2023NSCQ-LZX0126), the Science and Technology Research Program of Chongqing Municipal Education.}
\thanks{Liangbo Xie, Mude Cai and Xiaolong Yang are with the School of Communications and Information Engineering, Chongqing University of Posts and Telecommunications, Chongqing 400065, China (e-mail:xielb@cqupt.edu.cn; s240131108@stu.cqupt.edu.cn; yangxiaolong@cqupt.edu.cn).}
\thanks{Mu Zhou is with the School of Electronic Science and Engineering, Chongqing University of Posts and Telecommunications, Chongqing 400065, China, and the Chongqing Key Laboratory of Dedicated Quantum Computing and Quantum Artificial Intelligence, Chongqing 400065, China.}
\thanks{Jiacheng Wang and Dusit Niyato are with the College of Computing and Data Science, Nanyang Technological University, Singapore 639798 (e-mail: jiacheng.wang@ntu.edu.sg, dniyato@ntu.edu.sg).}}

\markboth{}%
{Shell \MakeLowercase{\textit{et al.}}: A Sample Article Using IEEEtran.cls for IEEE Journals}


\maketitle
\begin{abstract}Localization in mobile networks has been widely applied in many scenarios. However, an entity responsible for location estimation exposes both the target and anchors to potential location leakage at any time, creating serious security risks. Although existing studies have proposed privacy-preserving localization algorithms, they still face challenges of insufficient positioning accuracy and excessive communication overhead. In this article, we propose a privacy-preserving localization scheme, named PPLZN. PPLZN protects protects the location privacy of both the target and anchor nodes in crowdsourced localization. Specifically, PPLZN introduces a novel Zero-Sum Noise Generation (ZSNG) method based on homomorphic encryption, which is used to construct a zero-sum noise set without revealing any individual anchor's noise term. This establishes the foundation for subsequent noise-adding protection. To ensure privacy across all participating nodes, PPLZN employs the zero-sum mechanism that conceals location-related parameters by adding zero-sum noise while enabling accurate position estimation. Simultaneously, homomorphic encryption ensures that the target's estimated location remains confidential throughout the computation. Furthermore, to address the explosive increase in computational and communication costs when the number of anchors grows, we propose a Node Selection Algorithm (NSA). By evaluating the contribution degree of Geometric Dilution of Precision (GDOP), NSA selects high-quality anchors, thereby reducing the number of nodes involved in localization and improving scalability. Simulation results validate the effectiveness of PPLZN. Evidently, it can achieve accurate position estimation without location leakage and outperform state-of-the-art approaches in both positioning accuracy and communication overhead. In addition, PPLZN significantly reduces computational and communication overhead in large-scale deployments, making it well-fitted for practical privacy-preserving localization in resource-constrained networks.
\end{abstract}


\section{Introduction}

\IEEEPARstart{W}{ith} the increase in the demand for location sensing applications, location-based services (LBSs) with the help of intelligent devices are becoming increasingly concerned, providing high-quality services based on users’ locations \cite{1}. The basic premise of LBS is that users can obtain and upload precise and real-time locations using a location system. One of the most well-known positioning systems is the Global Navigation Satellite System (GNSS), which allows users to use transmitted signals to locate themselves from satellites orbiting the Earth. Although GNSS, such as Global Positioning System (GPS), is very effective, it may not always be available due to its limited coverage and signal blocking, especially in most indoor environments and some harsh outdoor environments\cite{2}. Fortunately, with the popularity and performance of smart terminal devices, ubiquitous terminal networks have become an alternative choice, allowing users to connect to nearby terminal reference points (also known as anchor points) to help analyze location-related parameters between them for self-positioning\cite{3}. In this framework, a variety of range-based positioning techniques can be used, including time of arrival (ToA), received signal strength (RSS), and time difference of arrival (TDoA) \cite{5}. The ranging-based collaborated localization is typically divided into three stages: anchor discovery, distance ranging, and location estimation\cite{13},\cite{39},\cite{1103}. First, the target must establish communication links with the anchors involved in the localization process. Second, the distances between the target and different anchors are measured. Third, the target's position is estimated based on the measured distances and the anchors' position. 

Although collaborative positioning can achieve good positioning performance, there is a security risk of exposing the nodes' position. This vulnerability arises because the anchors must share location-related information during the collaborative process, which malicious actors could exploit to infer precise positions through analysis of the exchanged data \cite{1-3},\cite{1-4}. Moreover, the anchors may not be limited to fixed base stations but often include mobile vehicles, drones, and other devices, which are typically unwilling to disclose location information publicly. Likewise, the targets prefer to keep their location results private and not shared with others\cite{1-5}. Location privacy is further threatened when a third-party server—rather than the user itself—performs the position estimation, creating opportunities for passive data leakage. Therefore, positioning methods for privacy protection have emerged one after another to solve the above problems \cite{7},\cite{8},\cite{9},\cite{10},\cite{11},\cite{24}.

In recent years, researchers have proposed a variety of schemes for privacy-preserving localization. The core privacy protection techniques in these schemes mainly include cryptographic methods, obfuscation strategies, and perturbation mechanisms. Halder and Newe developed a symmetric homomorphic encryption scheme to enable end-to-end encryption, ensuring data confidentiality against unauthorized access \cite{7}. Zeng et al. employed secret sharing to maintain mutual confidentiality of positions during information exchange \cite{8}, allowing the target's estimated location to remain hidden by encoding inputs into polynomials. Building on this, some studies combined homomorphic encryption to preserve location privacy not only among users but also from the server. Wang et al. proposed a privacy-preserving indoor localization scheme, DP3, which applies the exponential mechanism on the server side to generate perturbed noise that masks true locations, thereby simultaneously protecting both client-side location privacy and server-side data privacy \cite{11}. Li et al. introduced a method combining distance and angle measurements for single-anchor positioning; after anchors add noise to their results and transmit them to the user, more accurate positioning can be obtained \cite{12}. The follow-up work further enhanced this scheme by integrating anchor quality assessment \cite{13} and anchor-assisted mechanisms \cite{14}, thus improving positioning accuracy while maintaining the original privacy protection.

In addition, researchers have developed privacy-preserving localization schemes to minimize the impact on localization accuracy across different scenarios, often requiring a trade-off between security and efficiency. Consequently, several studies have explored the integration of multiple techniques to strengthen privacy protection. For example, Li and Sun employed secret sharing to securely collect measurement data and positions from participants, thereby preventing collectors from inferring private information through differential attacks \cite{16}. In Wi-Fi fingerprint-based localization models, Alikhani et al. proposed a method to protect user privacy against anonymity attacks by leveraging Hilbert curves and dual encryption to enhance privacy protection \cite{15}. Nieminen and Järvinen designed a hybrid approach combining homomorphic encryption with garbled circuits (GC) to address server-side security vulnerabilities \cite{17}. Zhang et al. further improved privacy protection by integrating private blockchain with a zero-sum noise injection mechanism \cite{18}.

Among the above approaches, homomorphic encryption and secret sharing provide strong privacy protection but require complex mathematical operations and encoding transformations between plaintext and ciphertext. As a result, directly applying homomorphic encryption to provide localization services leads to substantial computational and communication overhead. Differential Privacy (DP), a representative perturbation mechanism, preserves privacy by injecting suitable noise in dataset \cite{11}. However, the added noise inevitably reduces data utility and degrades localization accuracy, making DP unsuitable for scenarios demanding high positioning precision. In contrast, zero-sum noise achieves protection by ensuring that added-noise terms are canceled out during subsequent computations. This mechanism secures data without compromising positioning accuracy. However, if zero-sum noise is transmitted in plaintext, it is susceptible to intercepting or eavesdropping by malicious adversaries \cite{1021}.

These limitations highlight the need for a privacy-preserving localization algorithm that simultaneously ensures strong confidentiality and high positioning accuracy while avoiding excessive computational and communication costs. Motivated by this, we propose a novel privacy-preserving localization scheme, named PPLZN. We adopt the ToA algorithm for ranging in mobile networks. Within this framework, anchors operate under a mutually distrusted paradigm, and targets maintain adversarial suspicion toward the aggregator. Based on these premises, PPLZN ensures that the location privacy of every entity—including all anchors, the target, and the aggregator—is preserved, while the estimated position of the target remains known only to itself. The main contributions of this article are summarized as follows.
\begin{itemize}
    \item We propose a zero-sum noise generation mechanism based on Paillier homomorphic encryption, which integrates strong cryptographic confidentiality with noise cancellation that does not affect positioning accuracy. This design provides a methodological foundation for efficient privacy-preserving localization.

    \item We propose PPLZN, a complete privacy-preserving localization scheme. To improve scalability, we further design a Node Selection Algorithm (NSA) that reduces overhead in dense-anchor scenarios by selecting suitable anchors based on contribution of Geometric Dilution of Precision (GDOP), without leaking location-related information.
    
    \item We propose a rigorous security and performance evaluation framework, demonstrating through cryptographic proofs that the scheme achieves theoretical privacy guarantees under the defined model. Extensive simulations further confirm its practical advantages, showing reduced computational and communication overhead and improved efficiency in dense deployments.
\end{itemize}

The rest of this article is organized as follows. Section II introduces the system model, outlines the Paillier encryption scheme, and formulates the problem. Section III presents the design of the proposed PPLZN framework. Section IV provides a comprehensive performance analysis of PPLZN. Finally, Section V concludes the article.

\section{SYSTEM MODEL AND PROBLEM FORMULATION}
This section presents basic technical introductions, including the system model, conventional ToA localization, Paillier encryption schemes, and the problem formulation of privacy-preserving localization.
\subsection{System Model and Conventional ToA Localization}
We first describe three types of entities involved in the localization process \cite{9}, which are defined below.
\begin{itemize}
    \item {\textit{Target} $\mathbb{T}$}: The device needs to acquire its true position $\mathbf{p}_{0} = [x_{0}, y_{0}, z_{0}]^\text{T}$ through range measurements with the anchors, and we assume its estimated position $\hat{\mathbf{p}}_{0} = [\hat{x}_{0}, \hat{y}_{0}, \hat{z}_{0}]^\text{T}$.
    
    \item {\textit{Anchor} $\mathbb{A}_{i}$}: The anchors serve as reference entities for target positioning, typically possessing known positions denoted as $\mathbf{p}_{i} = [x_{i}, y_{i}, z_{i}]^\text{T}$. 
    
    \item {\textit{Aggregator(or Server)} $\mathbb{G}$}: The device is responsible for performing specific ciphertext calculations.
\end{itemize}

We consider a collaborative localization scheme within a network. In this scheme, the target has lost its position. Consequently, it transmits a localization request to neighboring anchors within its communication range. The anchors send location-related information back to the target. Using the information received from these anchors, the target can estimate its position\footnote{The main goal of this article is to explore a new privacy-preserving localization scheme from a theoretical point of view. Therefore, we do not consider the issues such as noise, non-line-of-sight (NLoS), and synchronization \cite{24} ,which can be left as the future work.}. We assume that the scheme contains one target and ${m}$ anchors \cite{9}. The estimated distance between the target $\mathbb{T}$ and the anchor $\mathbb{A}_{i}$ is indicated by ${d}_{i}$. Fig. \ref{fig1} shows the cooperative positioning scenario of network.
\begin{figure}[!t]
\centering
\includegraphics[width=2.8in]{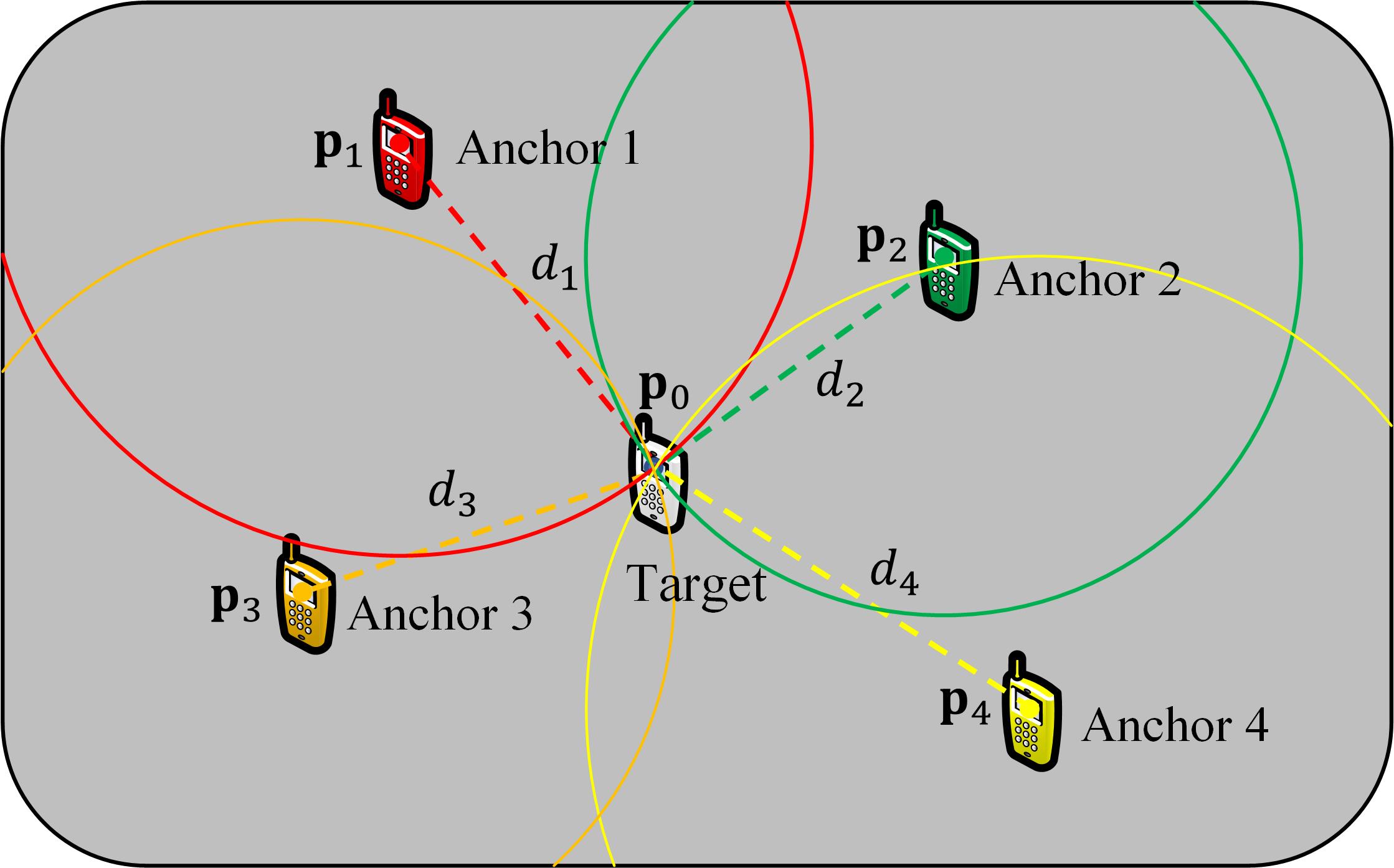}
\caption{Overview of cooperative positioning scenario. The target estimates distances to four surrounding anchors (${d}_{1}$-${d}_{4}$) and  utilizes a multilateration approach to acquire its own position $\mathbf{p}_{0}$ by leveraging location-related information from nearby mobile anchors.}
\label{fig1}
\end{figure}

The proposed ToA localization algorithm operates through two phases in the described scenario. During the range phase, the signal propagation time measurements yield distance estimates ${d}_{i}$ between the target and each anchor node ${i}$. Geometrically, each distance ${d}_{i}$ defines a spherical solution space centered at anchor coordinates $\mathbf{p}_{i}$ with radius ${d}_{i}$. The positioning phase subsequently resolves the target coordinates $\mathbf{p}_{0}$ using geometric intersection techniques, such as trilateral positioning \cite{19}. As depicted in Fig. \ref{fig1}, these spheres converge at a unique point under ideal conditions corresponding to the target location. 

In the above scenario, the estimated distance between the target $\mathbb{T}$ and the anchor $\mathbb{A}_{i}(i=1,2,\ldots,m)$ is expressed as
\begin{align}
d_{i}^{2} &= v^{2} (T_{i} - T_{0i})^{2} \nonumber\\
           &= \|\mathbf{p}_{i} - \mathbf{p}_{0}\|_{2}^{2} \nonumber\\
           &= x_{i}^{2} + y_{i}^{2} + z_{i}^{2} + x_{0}^{2} + y_{0}^{2} + z_{0}^{2}  \nonumber \\
           &\quad -2 (x_{i} x_{0} + y_{i} y_{0} + z_{i} z_{0}), \label{eq1}
\end{align}
where ${v}$ is the propagation speed of the signal (usually the speed of light). $T_{i}$ and $T_{0i}$ are the timestamps of the range signal received by the anchor $A_{i}$ and the timestamp of the range signal sent by the target $\mathbb{T}$, respectively. Let $R_i=x_i^2+y_i^2+z_i^2$, a set of distance equations is given as
\begin{equation}
\begin{cases}
  d_{1}^{2} - R_{1} &= -2(x_{1}x_{0} + y_{1}y_{0} + z_{1}z_{0}) + R_{0} \\
  d_{2}^{2} - R_{2} &= -2(x_{2}x_{0} + y_{2}y_{0} + z_{2}z_{0}) + R_{0} \\
  &\quad\vdots \\
  d_{m}^{2} - R_{m} &= -2(x_{m}x_{0} + y_{m}y_{0} + z_{m}z_{0}) + R_{0}. \label{eq2}
\end{cases} 
\end{equation}
\myref{eq2} can be converted into a matrix expression, given as
\begin{equation}
\mathbf{b} = \mathbf{A} \mathbf{x}, \label{eq3}
\end{equation}
where
\begin{equation}
\mathbf{b} = \left[ d_1^2 - R_1,\  d_2^2 - R_2,\  \ldots,\  d_m^2 - R_m \right]^{\mathrm{T}}, 
\label{eq4}
\end{equation}
\begin{equation}
\mathbf{A} = 
\begin{bmatrix} 
-2x_1 & -2y_1 & -2z_1 & 1 \\ 
-2x_2 & -2y_2 & -2z_2 & 1 \\ 
\vdots & \vdots & \vdots & \vdots \\ 
-2x_m & -2y_m & -2z_m & 1 
\end{bmatrix}, 
\label{eq5}
\end{equation}

\begin{equation}
\mathbf{x} = 
\begin{bmatrix} 
x_0 & y_0 & z_0 & R_0 
\end{bmatrix}^{\mathrm{T}}. 
\label{eq6}
\end{equation}
Assuming that the distance measurement noise is Gaussian noise, we can calculate the target position by minimizing the mean squared error (MMSE) between the true distance and the range distance as follows:
\begin{equation}
\mathbf{x} = \left( \mathbf{A}^{\mathrm{T}} \mathbf{A} \right)^{-1} \mathbf{A}^{\mathrm{T}} \mathbf{b}.
\label{eq7}
\end{equation}
Hence, the position of target is
\begin{equation}
\hat{\textbf{p}}_0 = 
\begin{bmatrix} 
\mathbf{x}(1), & \mathbf{x}(2), & \mathbf{x}(3) 
\end{bmatrix}^{\mathrm{T}}.
\label{eq8}
\end{equation}

\subsection{Paillier Encryption Scheme}
The Paillier cryptosystem is a partially homomorphic public key encryption scheme proposed by P. Paillier in 1999 \cite{20}. It comprises three core algorithms \cite{21}:
\begin{itemize}
    \item Key Generation: Select two large primes $p$ and $q$ of equal length satisfying 
    \begin{equation}
    \gcd(pq, (p-1)(q-1)) = 1,
    \label{eq9}
    \end{equation}
    where $\text{gcd}(\cdot, \cdot)$ is the largest common divisor of two natural numbers. Then, compute the RSA modulus $n=pq$ and $\lambda = \text{lcm}(p-1,q-1)$, where $\text{lcm}(\cdot, \cdot)$ computes the least common multiple of two integers. Next, select a random integer $g\in Z^*_{n^2}$ and compute
    \begin{equation}
    \alpha = \left( L\left( g^{\lambda}\mod n^2 \right) \right)^{-1} \mod n.
    \label{eq10}
    \end{equation}
    Finally, the public key $pk=(n,g)$, and the corresponding private key $sk=(\lambda,\alpha)$ \cite{22}.
    \item Encryption: Input plaintext $m\in Z_n$ and a random integer $r\in Z_n$, the ciphertext $c$ is computed as follows:
    \begin{equation}
    c = g^{m} \cdot r^{n} \mod n^{2}.
    \label{eq11}
    \end{equation}
    The Paillier encryption of $m$ is expressed by $\llbracket m\rrbracket_{pk}$.
    \item Decryption: Input ciphertext $c\in Z_{n^2}^*$. The corresponding plaintext is computed as 
    \begin{equation}
    m=L(c^{\lambda}\mod n^2)\cdot \alpha\mod n.
    \label{eq12}
    \end{equation}
    The Paillier decryption of $c$ is expressed by $Decr(c)$.
\end{itemize}

Paillier encryption, which has additive homomorphic properties \cite{23}, is central to our design. In this work, we denote multiplication between a plaintext and a ciphertext by $\otimes$, and addition between two ciphertexts by $\oplus$. The homomorphic addition and homomorphic scalar multiplication properties of Paillier encryption are exemplified by \myref{eq13} and \myref{eq14}, respectively.
\begin{align}
\forall m_{1},& m_{2} \in \mathbb{Z}_{n}  \text{ , } k \in \mathbb{N} \nonumber \\
\llbracket m_1 \rrbracket_{pk} \oplus \llbracket m_2 \rrbracket_{pk} &= \llbracket m_1 \rrbracket_{pk} \llbracket m_2 \rrbracket_{pk} \mod n^{2}  \nonumber\\
               &= \llbracket m_1 + m_2 \rrbracket_{pk} \mod n, \label{eq13}\\
 k\otimes \llbracket m_1\rrbracket _{pk}&= (\llbracket m_1 \rrbracket_{pk}^{k} \mod n^{2}). \label{eq14}
\end{align}
where $m_{1}$ and $m_{2}$ are two different plaintext messages in the integer ring of modulus $n$. $k$ is a natural number, $\llbracket \cdot \rrbracket_{pk}$ represents the ciphertext encrypted with public key $pk$, and $\llbracket \cdot \rrbracket_{pk}^{k}$ denotes the ciphertext raised to the $k$-th power.

From the homomorphic addition and homomorphic scalar multiplication properties of Paillier encryption, the multiplication between a plaintext matrix and a ciphertext vector can be performed as follows:
\begin{equation}
\mathbf{A} \otimes \llbracket \mathbf{m} \rrbracket_{pk} = 
\begin{bmatrix}
\mathbf{A}_{11} \otimes \llbracket \mathbf{m}(1) \rrbracket_{pk} \oplus \cdots \oplus \mathbf{A}_{1n} \otimes \llbracket \mathbf{m}(n) \rrbracket_{pk} \\
\vdots \\
\mathbf{A}_{n1} \otimes \llbracket \mathbf{m}(1) \rrbracket_{pk} \oplus \cdots \oplus \mathbf{A}_{nn} \otimes \llbracket \mathbf{m}(n) \rrbracket_{pk}
\end{bmatrix},
\label{15}
\end{equation}
where $\mathbf{A} \in \mathbb{Z}^{n \times n}$ and $\mathbf{m} \in \mathbb{Z}^{n}$.

\subsection{Geometric Dilution of Precision}
GDOP has been proposed to evaluate the influence of the geometric distribution of anchors on positioning performance \cite{25}. In general, larger GDOP values correspond to greater positioning errors \cite{26}. In line-of-sight environments, all the measurement errors can be considered to be zero-mean independent and identically distributed Gaussian variables in ToA positioning systems \cite{27}. \myref{eq1} is differentiated into
\begin{equation}
\frac{x_0-x_i}{d_i} d(x_0) + \frac{y_0-y_i}{d_i} d(y_0) + \frac{z_0-z_i}{d_i} d(z_0) = d(d_i). \label{eq16}
\end{equation}
These equations can be represented in matrix form as
\begin{equation}
\mathbf{H}_m d\mathbf{p} = d\mathbf{D}, \label{eq17}
\end{equation}
where
\begin{equation}
\mathbf{H}_m = 
\begin{bmatrix}
\frac{x_0-x_1}{d_1} & \frac{y_0-y_1}{d_1} & \frac{z_0-z_1}{d_1} \\
\vdots & \vdots & \vdots \\
\frac{x_0-x_m}{d_m} & \frac{y_0-y_m}{d_m} & \frac{z_0-z_m}{d_m}
\end{bmatrix}, \label{eq18}
\end{equation}
\begin{equation}
d\mathbf{p} = 
\begin{bmatrix}
d x_0 \\ 
d y_0 \\ 
d z_0 \label{eq19}
\end{bmatrix}, 
\end{equation}
\begin{equation}
d\mathbf{D} = 
\begin{bmatrix}
d(d_1) \\ 
\vdots \\ 
d(d_m)  \label{eq20}
\end{bmatrix}. 
\end{equation}
GDOP is defined as
\begin{equation}
\text{GDOP}_m = \sqrt{\operatorname{trace}(\mathbf G_m)}, \quad 
\mathbf G_m = (\mathbf{H}_m^\text{T} \mathbf{H}_m)^{-1}. \label{eq21}
\end{equation}
where $\operatorname{trace}(\cdot)$ represents the trace of the matrix, and $\mathbf{H}_{m}$ is the observation matrix with $m$ anchors.

To minimize GDOP, a traversal search can be applied. However, because GDOP computation requires matrix inversion and multiplication, the computational cost grows rapidly with the number of anchors. To address this issue, we introduce a reverse star selection algorithm, which efficiently identifies the anchor set with the optimal contribution to GDOP. The definition of contribution degree is derived as follows. The measurement matrix can be expressed as 
\begin{equation}
\mathbf{H}_{m} = \left[ \mathbf{H}_{m-1}^\text{T} ,\mathbf{h}_{i}^\text{T} \right]^\text{T},
\label{eq22}
\end{equation}
where $\mathbf{h}_{i}$ is the row vector of the observation matrix corresponding to the excluded anchor $\mathbb A_{i}$. So, the GDOP is rewritten as
\begin{equation}
\text{GDOP}_{m}^{2} = \operatorname{trace} \left( \mathbf{H}_{m-1}^\text{T} \mathbf{H}_{m-1} + \mathbf{h}_{i}^\text{T} \mathbf{h}_{i} \right)^{-1}.
\label{eq23}
\end{equation}
It can be further expressed as
\begin{equation}
\Delta \text{GDOP}_{i}^{2} = \text{GDOP}_{m-1}^{2} - \text{GDOP}_{m}^{2} = \operatorname{trace} \left( \frac{ \mathbf{G}_{m} \mathbf{h}_{i}^\text{T} \mathbf{h}_{i} \mathbf{G}_{m} }{1 - \mathbf{h}_{i} \mathbf{G}_{m} \mathbf{h}_{i}^\text{T} } \right) \label{eq24}
\end{equation}
where $\Delta \text{GDOP}_{i}^{2}$ denotes the contribution degree of the anchor $\mathbb A_{i}$, representing the change in the squared GDOP value of the set when the anchor is excluded. Therefore, the contribution degree of GDOP provides a quantitative measure for assessing the effect of an anchor on positioning accuracy. In this work, we investigate anchor node selection strategies based on the contribution degree of GDOP.

\subsection{Privacy-Preserving Localization}
Following the common practice in privacy-preserving localization studies \cite{29},\cite{34}, this work adopts the honest-but-curious model. All participating entities execute the protocol faithfully but may attempt to infer private information from others. We  consider the case where the target issues a single localization request, and the anchors respond accordingly.
Our goal is to design a privacy-preserving ToA-based localization scheme that simultaneously guarantee practical utility requirements:
\begin{itemize}
    \item Target Location: The target can calculate its position using the PPLZN scheme, which is equivalent to the MMSE estimation in \myref{eq8}.
    \item Privacy Preservation: For the target $\mathbb{T}$ and any anchor $\mathbb{A}_i(1,2,\ldots,m)$, the location of the target and the anchors cannot be estimated by others.
    \item High Efficiency: To ensure applicability in real-world scenarios—especially under high anchor density—performance analysis demonstrates that our algorithm achieves lower computational and communication overhead compared to those of existing privacy-preserving localization schemes.
\end{itemize}

\section{DESIGN OF PPLZN}
In this section, we present PPLZN, a scheme designed to ensure that no location-related information of any node is disclosed during the localization process. As illustrated in Fig. \ref{fig_frame}, the scheme consists of three modules: zero-sum noise generation module, NSA module, and ToA-based privacy-preserving localization module. The following will explain in detail the principles of each module and the system framework.
\begin{figure*}[!t]
\centering
\includegraphics[width=7in]{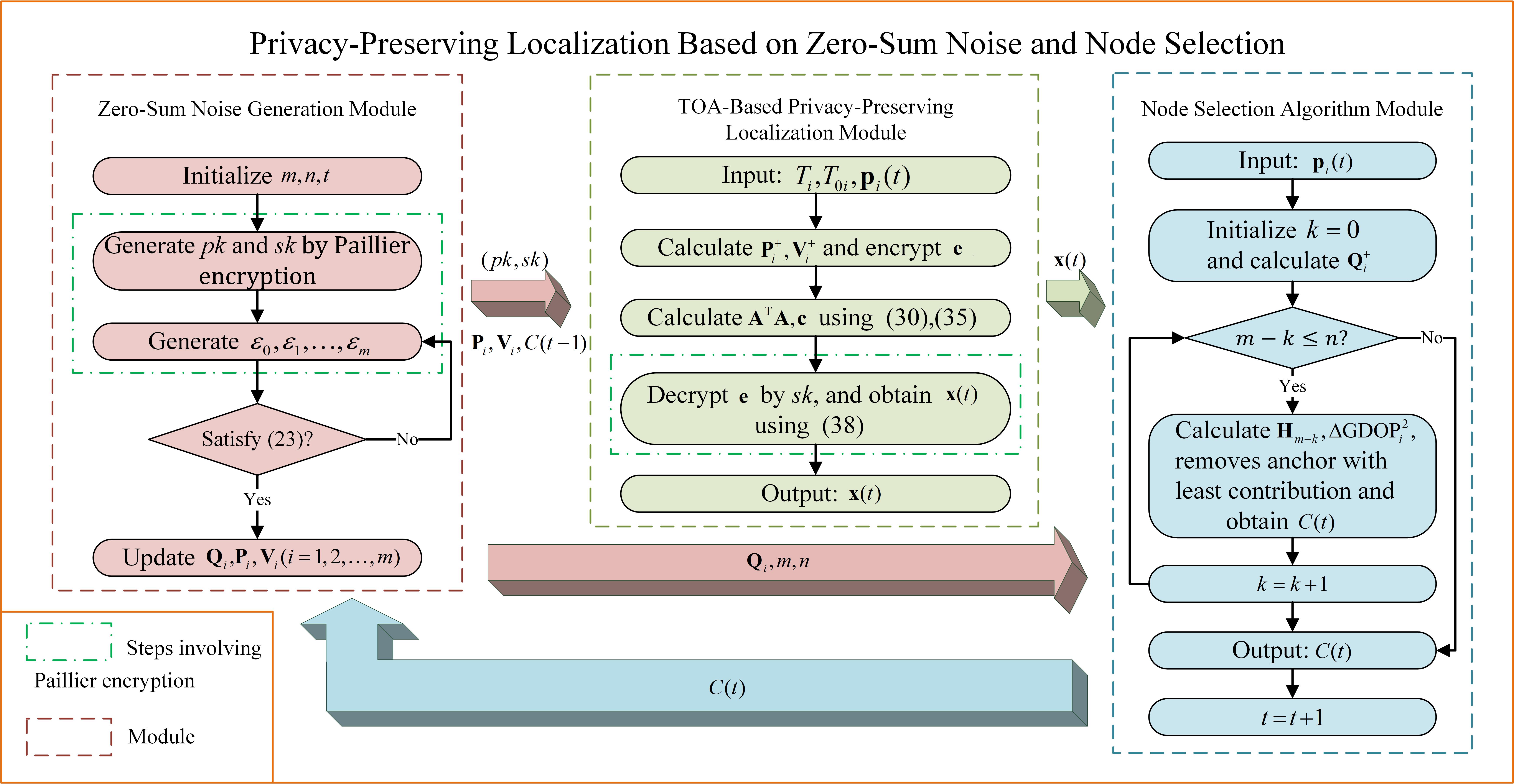}
\caption{Algorithm framework of PPLZN. The framework comprises three modules: zero-sum noise generation module (left) to protect privacy, ToA-based privacy-preserving localization module (middle) to estimate the target's position, and node selection algorithm module (right) to improve computation efficiency. The zero-sum noise generation module encrypts location-related information, then the localization module decrypts the ciphertext and provides estimated positions to the node selection algorithm module, and finally the optimal anchors are obtained for the next iteration.}
\label{fig_frame}
\end{figure*}

\subsection{Zero-Sum Noise Generation Based on Paillier Encryption}
The first step of PPLZN is the generation of a zero-sum noise set. This injected noise simultaneously protects sensitive location data and preserves localization accuracy, as the noise terms cancel out during position computation \cite{1021}. Based on the scenario model described in Section II, the model can be formally expressed as

\begin{equation}
\sum_{i=0}^m \varepsilon_i = 0, \label{eq1023}
\end{equation}
or
\begin{equation}
\varepsilon_0 = -\sum_{i=1}^m \varepsilon_i, \label{eq1024}
\end{equation}
where $\varepsilon_{0}$ and $\varepsilon_{i}(i=1,2,\cdots, m)$ denote the random noise generated by the target $\mathbb T$ and the anchor $\mathbb A_{i}$, respectively.

Specifically, the anchor $\mathbb A_{i}$ generates a random number $\varepsilon_{i}$ locally and transmits it to the target $\mathbb{T}$. After collecting the random numbers $\varepsilon_{1},\varepsilon_{2},\ldots,\varepsilon_{m}$ from all the anchors, the target calculates $\varepsilon_{0}$ using \myref{eq1024}. The $\varepsilon_{0},\varepsilon_{1},\varepsilon_{2},\cdots,\varepsilon_{m}$ make up a set of zero-sum noise. However, directly transmitting these noise values may expose anchor data to the target. To reduce and avoid this risk, we design a method for zero-sum noise generation based on Paillier encryption.

The Paillier encryption-based zero-sum noise generation process is illustrated in Fig. \ref{fig2}. First, the target generates a public-private key pair ${(pk, sk)}$ using the Paillier encryption, where $pk$ is the public key, and $sk$ is the corresponding private key. After generating the keys, the target distributes the public key to all anchors. Each anchor $\mathbb A_{i}$ encrypts its locally generated random number $\varepsilon_{i}$ with $pk$ and transmits the resulting ciphertext to the aggregator $\mathbb{G}$. The aggregator collects all the ciphertexts from the anchors and computes $\sum\varepsilon_{i}$ using the Paillier homomorphic addition property \myref{eq13}. The sum is then sent to the target. Importantly, ciphertexts must be transmitted to the aggregator rather than directly to the target; otherwise, the target could decrypt the noise of each anchor individually, leaking sensitive information in subsequent noise injection steps. Finally, the target decrypts it using $sk$ to obtain the noise sum and derive the local zero-sum noise component by \myref{eq1024}. So far, a complete set of zero-sum noise components $\varepsilon_{0},\varepsilon_{1},\varepsilon_{2},\ldots,\varepsilon_{m}$ has been generated.

Based on a single set of zero-sum noise, multiple sets are combined to form a zero-sum noise matrix, denoted as $\mathbf P_{i}(i=0,1,\ldots m)$, satisfying 
\begin{equation}
\sum_{i=0}^{m}\mathbf P_{i}=\mathbf 0.  \label{eq1127}
\end{equation}

\begin{figure}
\centering
\includegraphics[width=3in]{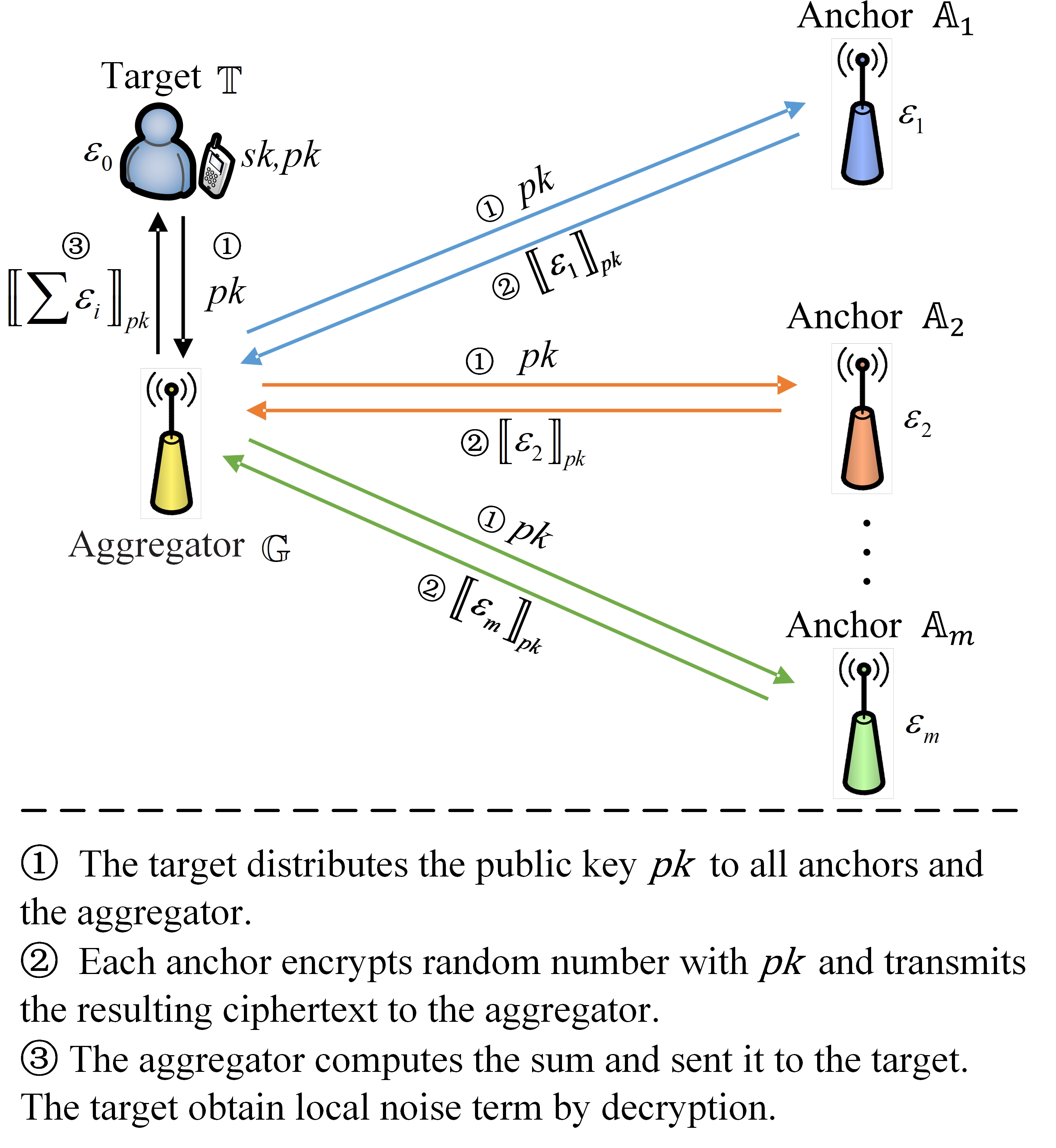}
\caption{Zero-sum noise generation based on Paillier encryption.}
\label{fig2}
\end{figure}

\subsection{Node Selection Algorithm Based on Contribution Degree of GDOP}
Multiliterate positioning accuracy exhibits a positive correlation with the number of participating anchors, asymptotically approaching theoretical limits under constant measurement noise \cite{CRLB_TOA}. However, computational and communication overhead increases explosively with anchor count. To combine efficiency with positioning precision in dense anchor environments, this section presents a privacy-preserving implementation of the NSA inspired by the reverse star selection algorithm \cite{28}. 

Based on \myref{eq24}, the privacy matrix $\mathbf{G}_m$ and $h_i$ must be protected. Since $\mathbf G_{m}$ is made up of $h_{i}(i=1,2,\ldots, m)$, only $h_{i}$ needs to be protected. The privacy-preserving methodology proceeds as follow:

\begin{enumerate}
\item{Initialize the optimal anchor combination, which contains the identities of all anchors. $\mathbf Q_{i}$, $i=1,2,\ldots, m$ is a set of zero-sum matrices generated by the ZSNG, satisfying
\begin{equation}
\sum_{i=0}^{m}\mathbf Q_{i} = \mathbf 0 . \label{eq1128}
\end{equation}
where $\mathbf Q_{i}$ has the same size as its corresponding $h_{i}$.}

\item{The aggregator broadcasts the GDOP calculation request. The anchor $\mathbb A_{i}$ sends $\mathbf Q_{i}^{+}$, the target $\mathbb  T$ sends $\mathbf Q_{0}^{+}$, and the other anchor $\mathbb A_{j}(j\neq i)$ sends $\mathbf Q_{j}$, where 
\begin{equation}
\begin{cases}
  \mathbf Q_{i}^{+} = \mathbf Q_{i}-\mathbf p_{i} & i \ne0\\
  \mathbf Q_{0}^{+} = \mathbf Q_{0}+ \widehat{\mathbf{p}}_{0} & i=0 
    \label{eq1132}.
\end{cases}
\end{equation}

The aggregator collects all return signals and calculates $\mathbf H_{m}$ from \myref{eq25} and \myref{eq26}}. $\Delta \text{GDOP}_{i}^{2}$ is calculated by \myref{eq24}. 

\item{The aggregator removes the identity of the anchors with the least contribution from the optimal anchor combination and builds a new observation matrix $\mathbf H_{m-k}$ until the number of anchors in the optimal anchor combination is equal to $n$.}
\end{enumerate}

\begin{equation}
\textbf{H}_m = \left[ \begin{array}{c} h_1 \\ \vdots \\ h_m \end{array} \right], \label{eq25}
\end{equation}

\begin{align}
\hat{h}_i &= \frac{\mathbf Q_i^{+} +\mathbf Q_0^{+} + \sum_{j=1, j\neq i}^{m} \mathbf Q_j}{\left| \mathbf Q_i^{+} + \mathbf Q_0^{+} + \sum_{j=1, j\neq i}^{m} \mathbf Q_j \right|} \nonumber \\
    &= \frac{\hat{p}_0 - p_i}{\left| \hat{p}_0 - p_i \right|} = \left[ \frac{\hat{x}_0 - x_i}{d_i}, \frac{\hat{y}_0 - y_i}{d_i}, \frac{\hat{z}_0 - z_i}{d_i} \right]^{T}. \label{eq26}
\end{align}

According to the above, we select $n$ anchors from the set of $m$ anchors at the moment of $t$ (the NSA is not executed when $m \le n$). These $n$ anchors are considered to be the optimal anchor combination, denoted $C(t)$. The NSA workflow executed by the aggregator is shown in Algorithm \ref{alg2}.

\begin{algorithm}
\renewcommand{\algorithmicrequire}{\textbf{Input:}} 
\renewcommand{\algorithmicensure}{\textbf{Output:}}
\pagestyle{empty}
\caption{Node Selection Algorithm (NSA)}
\begin{algorithmic}[1]
\REQUIRE anchor positions $\mathbf{p}_i(t)$'s; public key $(n,g)$; target estimated position at the moment of $t$ $\mathbf{\hat p}_0(t)$
\ENSURE the optimal anchor combination at the moment of $t$ $C(t)$
\STATE Initialize $C(t)$ and generate $\mathbf{Q}_i$
\FOR{each anchor $i$}
    \STATE Construct $\mathbf{Q}_i^{+}$ and send it to aggregator
\ENDFOR
\STATE Target constructs $\mathbf{Q}_0$ and sends it to aggregator
\WHILE{$m-k \leq n$}
    \STATE Aggregator calculates $\mathbf{H}_{m-k}$,$\Delta \mathrm{GDOP}_i^2$ and removes anchor with least contribution from $C(t)$
    \STATE $k \leftarrow k + 1$ (initial $k=0$)
\ENDWHILE
\STATE Target obtains $C(t)$ from aggregator
\STATE $t \leftarrow t + 1$
\end{algorithmic} \label{alg2}
\end{algorithm}

\subsection{ToA-Based Privacy-Preserving Localization}
We now consider the location estimation formula in \myref{eq7}, i.e. $\mathbf{x}=\left(\mathbf{A}^\text{T}\mathbf A\right)^{-1}\mathbf A^\text{T}\mathbf b$, where $\mathbf A$ is defined by the anchor coordinates given in \myref{eq5}. Thus, computing the target’s position requires both the coordinates of each anchor and the measured ranges between the target and the anchors. If relevant information is sent without any privacy-preserving measures, the anchor location will be exposed to the target, which may violate the privacy goal.

In general, we perform multiple decompositions of the position estimation and use different encryption methods for each decomposition term, as shown in Fig. \ref{fig_decom}. The principles of decomposition and the corresponding encryption approaches are analyzed in detail below.
\begin{figure}
\centering
\includegraphics[width=3in]{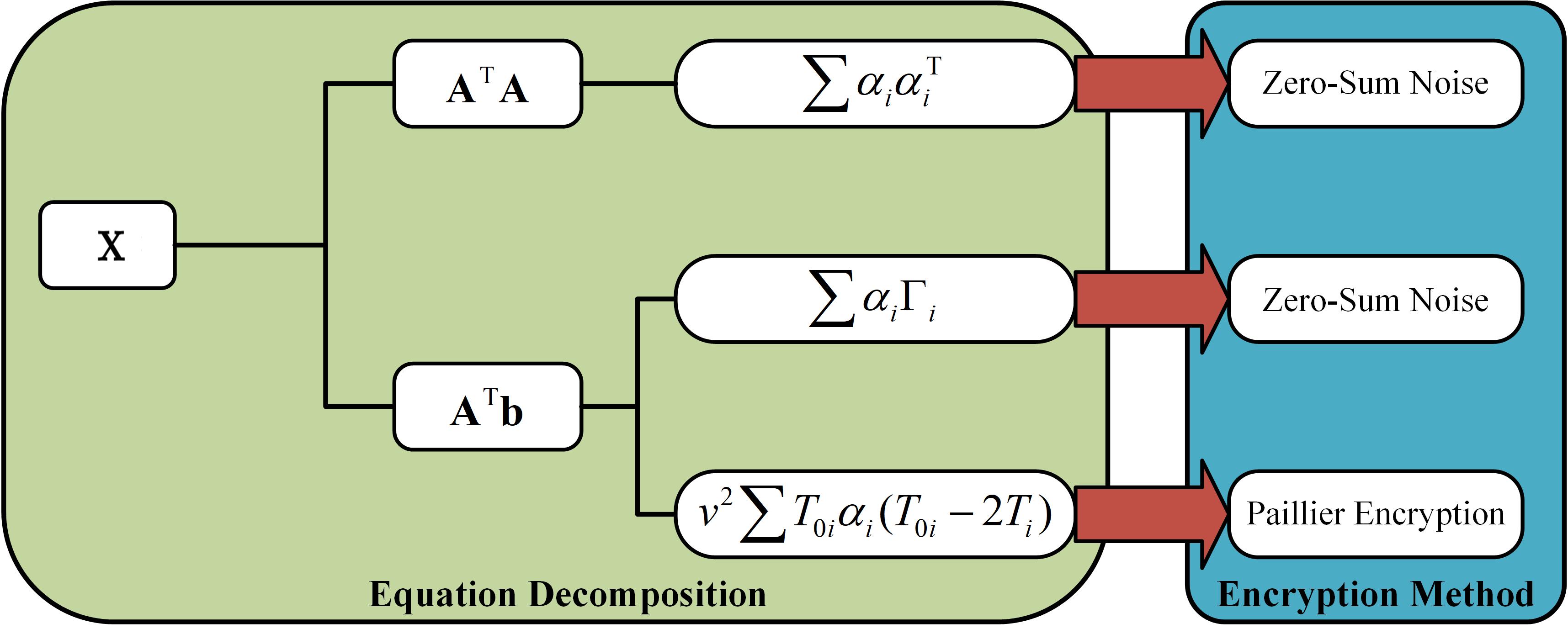}
\caption{Equation decomposition and encryption method. $\mathbf x$ is decomposed into two matrices based on (7) and further transformed into three summation expressions, each protected with a suitable encryption method.}
\label{fig_decom}
\end{figure}

To meet the privacy requirement, we decompose $\mathbf{x}$ into two steps, $\mathbf A^\text{T}\mathbf A$ and $\mathbf A^\text{T}\mathbf b$. According to \myref{eq4} and \myref{eq5}, let $\alpha_{i}=\left[-2 x_{i}, -2 y_{i}, -2 z_{i}, 1\right]^\text{T}$, $b_{i}=d_{i}^{2}-R_{i}$, then $\mathbf A$ can be written as

\begin{equation}
\mathbf A = \left[ \begin{array}{c} \alpha_1^\text{T} \\ \alpha_2^\text{T} \\ \vdots \\ \alpha_m^\text{T} \end{array} \right]. \label{eq27}
\end{equation}
Furthermore,
\begin{align}
\mathbf A^\text{T}\mathbf A &= \sum_{i=1}^m \alpha_i \alpha_i^\text{T}, \label{eq28} \\
\mathbf A^\text{T}\mathbf b &= \sum_{i=1}^m \alpha_i b_i. \label{eq29}
\end{align}

The anchor $\mathbb A_{i}$ holds its private parameters $\alpha_{i}$ and $b_{i}$ in \myref{eq28} and \myref{eq29}. Accordingly, each anchor can locally compute the terms $\alpha_{i}\alpha_{i}^\text{T}$ and $\alpha_{i} b_{i}$. Since the target $\mathbb{T}$ requires only the summations $\sum\alpha_{i}\alpha_{i}^{T}$ and $\sum\alpha_{i} b_{i}$, this enables the use of zero-sum noise. 
To compute $\mathbf A^\text{T}\mathbf A$ securely, each anchor $\mathbb A_i$ generates its private matrix  $\alpha_i\alpha_i^{T}(i=1,2,\ldots, m)$,  while the target seeks to obtain the sum $\sum_{i=1}^{m}\alpha_{i}\alpha_{i}^{T}$  without revealing any individual anchor position. The procedure for securely computing $\mathbf A^\text{T}\mathbf A$ using zero-sum noise is as follows.

\begin{enumerate}
\item{Based on the ZSNG, anchor $\mathbb A_{i}$ derives a zero-sum noise matrix $\mathbf P_{i}$, satisfying \myref{eq1127}, where $\mathbf P_{i}$ has the same size as its corresponding $\alpha_{i}\alpha_{i}^{T}$.}

\item{The anchor computes the noise-adding information $\mathbf P_{i}^{+}$ and transmits it to the target, where 
\begin{equation}
\mathbf P_{i}^{+} = \mathbf P_{i} + \alpha_{i}\alpha_{i}^{T}. \label{eq1131}
\end{equation}
}

\item{Utilizing its own locally generated matrix $\mathbf P_{0}$, the target aggregates all received matrices and computes the global summation using \myref{eq30}.}
\end{enumerate}

\begin{align}
\mathbf A^\text{T}\mathbf A &= \sum_{i=1}^{m} \mathbf P_i^{+} + \mathbf P_0 = \sum_{j=1}^{m} (\alpha_j\alpha_j^{T} + \mathbf P_j) + \mathbf P_0 \nonumber\\ &= \sum_{j=1}^{m} \alpha_j\alpha_j^{T} + \sum_{j=0}^{m} \mathbf P_j = \sum_{j=1}^{m} \alpha_j\alpha_j^{T}. \label{eq30}
\end{align}

To protect $\mathbf A^\text{T}\mathbf b$, note that each $b_{i}$ contains two timestamps, $T_{i}$ and $T_{0 i}$ , from the anchors and the target, respectively. Directly sharing these timestamps would reveal the distance between the anchor and the target. If anchors gain access to the target's timestamps, malicious anchors could estimate the target's position through multilateration. Conversely, if different targets measure distances to the same anchor, they could infer the anchor's position. Therefore, access to raw timestamp data must be strictly prohibited for all participating entities to preserve location confidentiality in our system. Based on \myref{eq4}, $\mathbf b$ is expressed as:

\begin{equation}
\mathbf b = \left[ \begin{array}{c} b_{1} \\ b_{2} \\ \vdots \\ b_{m} \end{array} \right] = \left[ \begin{array}{c} v^{2} T_{01}^{2} - 2 v^{2} T_{1} T_{01} + \Gamma_{1} \\ v^{2} T_{02}^{2} - 2 v^{2} T_{2} T_{02} + \Gamma_{2} \\ \vdots \\ v^{2} T_{0 m}^{2} - 2 v^{2} T_{m} T_{0 m} + \Gamma_{m} \end{array} \right], \label{eq31}
\end{equation}
\begin{align}
b_i &= v^2 T_{0 i}^2 - 2 v^2 T_i T_{0 i} + \Gamma_i, \label{eq32} \\
\Gamma_i &= v^2 T_i^2 - (x_i^2 + y_i^2 + z_i^2). \label{eq33}
\end{align}
Note that the sending and receiving times of transmitted signals are kept by the target and the anchor side separately, which is confidential from each other. By focusing on $\mathbf A^\text{T}\mathbf b$, since  $\alpha_{i}, T_{i}$ and $\Gamma_{i}$ belong to the anchor, while $T_{0 i}$ is held by the target, we divide $\mathbf A^\text{T}\mathbf b$ into two parts as
\begin{align}
\mathbf A^\text{T}\mathbf b &= \sum_{i=1}^m \alpha_i (v^2 T_{0 i}^2 - 2 v^2 T_i T_{0 i} + \Gamma_i)\nonumber \\
&= \sum_{i=1}^m \alpha_i \Gamma_i + v^2 \sum_{i=1}^m T_{0 i} \alpha_i (T_{0 i} - 2 T_i).\label{eq34}
\end{align}

Let $\mathbf c = \sum_{i=1}^{m}\alpha_{i}\Gamma_{i}$. The privacy-preserving computation for $\mathbf c$ can be achieved using zero-sum noise. Suppose $\mathbf V_{i}(i=1,2,\ldots, m)$ is a set of random matrices generated by the ZSNG that satisfies
\begin{equation}
\sum_{i=0}^{m} \mathbf V_{i}=\mathbf{0}, \label{eq1129}
\end{equation}
where $\mathbf V_{i}$ has the same size as its corresponding $\alpha_{i} b_{i}$. Thus, the noise-added matrix can be defined as
\begin{equation}
\mathbf V_{i}^{+} = \alpha_{i}\Gamma_{i} + \mathbf V_{i}. \label{eq1130}
\end{equation}
And $\mathbf c$ can be calculated by \myref{eq35} in a way similar to the calculation of $\mathbf A^\text{T}\mathbf A$.
\begin{equation}
\mathbf c = \sum_{i=1}^m V_i^{+} + V_0 =  \sum_{i=1}^m \alpha_i \Gamma_i . \label{eq35}
\end{equation}

For $v^2 \sum_{i=1}^m T_{0 i} \alpha_i (T_{0 i} - 2 T_i)$ in \myref{eq34}, as $v^2$ is a public known quantity, it does not require encryption; only $\sum_{i=1}^m T_{0 i} \alpha_i (T_{0 i} - 2 T_i)$ needs to be considered. The target encrypts $T_{0 i}$ with a public key (this ciphertext is denoted as $t_{0 i}$) and the anchor encrypts $-2T_{i}$ with a public key (this ciphertext is denoted as $t_{i i}$). After the target sends it to anchor $\mathbb A_{i}$, the anchor calculates $\chi_{i}$ using the homomorphic property by \myref{eq36}.
\begin{equation}
\chi_i = \alpha_i \otimes (t_{0 i} \oplus t_{i i}). \label{eq36}
\end{equation}
 Until the aggregator $\mathbb{G}$ receives $\chi_{i}(i=1,2,\ldots, m)$ from all
anchors and $T_{0 i}(i=1,2,\ldots, m)$ from the target, it does not
compute $\mathbf{e}$ by \myref{eq37}.
\begin{equation}
\mathbf e = (T_{01} \otimes \chi_1) \oplus (T_{02} \otimes \chi_2) \oplus \cdots \oplus (T_{0 m} \otimes \chi_m). \label{eq37}
\end{equation}

In conclusion, the target receives a single values $\mathbf{e}$ from the aggregatorand subsequently  estimates its position using \myref{eq38}.
\begin{equation}
\mathbf x = (\mathbf A^\text{T}\mathbf A)^{-1} (\mathbf c + v^{2} \cdot {Decr}(\mathbf e)). \label{eq38}
\end{equation}

\begin{proposition} \label{prop1}
The proposed PPLZN computational procedure ensures that the locations of both the target and the anchors cannot be inferred by any other party.
\end{proposition} 
\begin{proof} 
For anchor-to-anchor communication, no location information is disclosed since no data is exchanged between anchors.

For target-to-anchor, take anchor $\mathbb A_{i}$ as an example without loss of generality. Anchor $\mathbb A_{i}$ sends $\mathbf P_{i}^{+}$ and $\mathbf V_{i}^{+}$ to the target, both obfuscated by the noise. Moreover, $\mathbb A_{i}$ only receives $t_{0 i}$ from the target, which is encrypted by Paillier encryption. Therefore, anchor $\mathbb A_{i}$ can disclose nothing about target’s location information and vice versa.

For node-to-third party, the aggregator receives $\mathbf Q_{i}^{+},\chi_i$ from each anchor and the target. $\mathbf Q_{i}^{+}$ is obfuscated by the noise and $\chi_i$ is encrypted by Paillier encryption. However, position estimation is performed on the target-side. The estimated position is decrypted by the target, leveraging the Paillier encryption scheme. The Paillier encryption scheme relies on the Decisional Composite Residuosity assumption \cite{36}, which posits that determining residuosity classes modulo a composite number is computationally infeasible. So, the aggregator cannot know the location information of any anchor or the target during the localization process.
\end{proof}
\begin{proposition} \label{prop2}
The estimated position result with the proposed zero-sum noise and Paillier encryption strategies is consistent with that without encryption.
\end{proposition} 
\begin{proof} 
Refer to Appendix A
\end{proof}
Finally, the process of privacy-preserving localization based on ToA is summarized in Algorithm \ref{alg1}.

\begin{algorithm}
\renewcommand{\algorithmicrequire}{\textbf{Input:}} 
\renewcommand{\algorithmicensure}{\textbf{Output:}}
\pagestyle{empty}
\caption{ToA-Based Privacy-Preserving Localization}
\begin{algorithmic}[1]
\REQUIRE Anchor positions $\mathbf{p}_i(t)$'s; transmitted timestamps $T_i(t)$'s; received timestamps $T_{0i}(t)$'s; public key $(n,g)$; optimal combination at the moment of $t-1$ $C(t-1)$
\ENSURE Target position at the moment of $t$ $\mathbf{p}_0(t)$
\STATE Generate $\mathbf{P}_i, \mathbf{V}_i$ by (27) and (41)
\STATE Target encrypt $t_{0 i}$ and send $T_{0i}$ to aggregator
\FOR{each anchor $i$ in $C(t-1)$}
    \STATE Construct $\mathbf{P}_i^{+}, \mathbf{V}_i^{+}$ by (35) and (42)
    \STATE Encrypt $t_{i i}$ using public key and compute $\chi_i$ by (44)
    \STATE Send $\mathbf{P}_i^{+}, \mathbf{V}_i^{+}$ to target and $\chi_i$ to aggregator
\ENDFOR
\STATE Aggregator computes $\mathbf e$ by (45) and sends it to target
\STATE Target computes $\mathbf A^\text{T}\mathbf A$, $\mathbf c$, decrypts $\mathbf e$ and computes $\mathbf{x}(t)$ by (46)
\STATE $t \leftarrow t + 1$
\end{algorithmic} \label{alg1}
\end{algorithm}

\subsection{Algorithm Framework}
In Fig. \ref{fig_frame}, the system operates through an integrated workflow involving three modules.
At the moment of $t$, the process begins with the Zero-Sum Noise Generation module, which produces multiple sets of zero-sum noise values $\varepsilon_0,\varepsilon_1, \ldots \varepsilon_m$ by a public-private key pair $(pk, sk)$. These noise values are used to update the noise matrices $\mathbf P_i, \mathbf Q_i$ and $\mathbf V_i$ by \myref{eq1127}, \myref{eq1128} and \myref{eq1129}. The module then outputs $(pk, sk)$, $\mathbf Q_i,\mathbf  P_i$, and the optimal anchor combination $C(t-1)$ to the Privacy-Preserving Localization module, while also transmitting $\mathbf Q_i, m$ and $n$ to the NSA module.
The Privacy-Preserving Localization module estimates the target's position based on the optimal anchors in $C(t-1)$. It computes the noise-added private matrices $\mathbf P_i^+$ and $\mathbf V_i^+$ and encrypts the vector $\mathbf e$. After decrypting the vector, the module estimates the target's location $\mathbf x(t)$, which is then sent to the NSA module. 
The NSA module accepts the position of anchors $\mathbf p_i(t)$ and $\mathbf x(t)$ as inputs and computes the GDOP contribution of each anchor $\Delta \text{GDOP}_{i}^{2}$ by \myref{eq23}. Then iteratively removes the anchor with the lowest contribution, recalculates the contributions, and repeats the process until the number of anchors is reduced from $m$ to $n$. The updated optimal set of anchors is returned to the Zero-Sum Noise Generation module to guide subsequent positioning rounds. Thus, the system achieves closed-loop operation on the basis of the above description. Finally, a schematic diagram of the information exchange between each entity for the PPLZN scheme is shown in Fig. \ref{fig_flowchart}.

\begin{figure}
\centering
\includegraphics[width=3in]{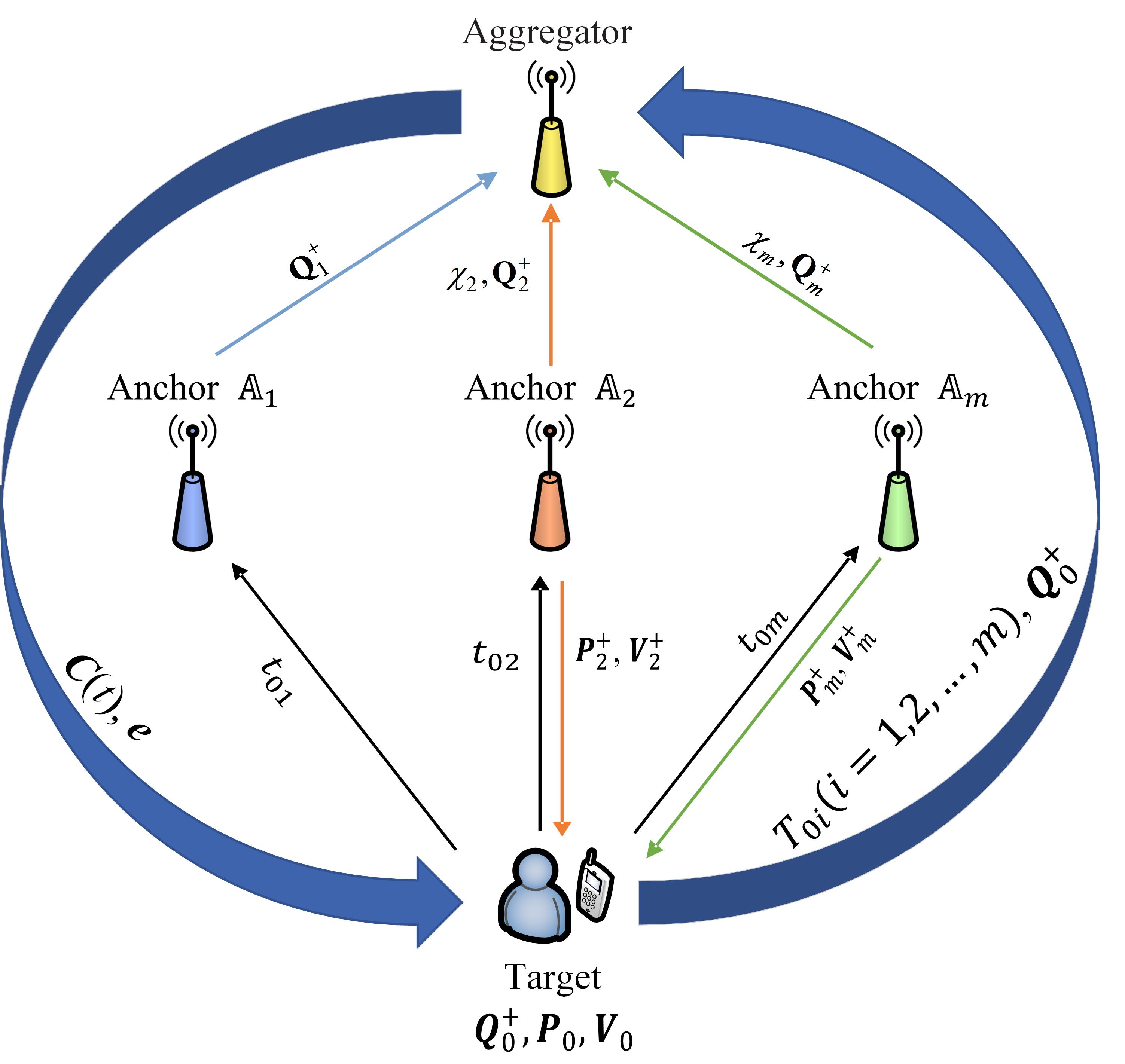}
\caption{Flowchart to the proposed privacy-preserving localization algorithm, where anchor 1 is removed from the node selection algorithm.}
\label{fig_flowchart}
\end{figure}

\subsection{Computation Complexity Analysis}
The computational complexity of the PPLZN scheme mainly comes from Paillier homomorphic encryption operations and matrix operations. Assume that the spatial dimension of the target be $d$,  the number of anchor nodes be $m$, and the key length of Paillier be $k$ bits. The time complexity of a single Paillier encryption or decryption operation is $\mathcal{O}(k^3)$. In ZSNG, each anchor node needs to perform Paillier encryption once to generate local noise. The total encryption overhead brought by $m$ anchor nodes is $\mathcal{O}(mk^3)$. After the aggregator collects all the ciphertexts, it performs homomorphic addition attribute calculations, requiring $m-1$ homomorphic additions, with a cost of $\mathcal{O}((m-1)k^2)$. After the target receives the aggregated value, it only needs to perform Paillier decryption once. This module has a complexity of $\mathcal{O}(k^3)$. The total complexity of the ZSNG module is $\mathcal{O}(mk^3)$. NSA calculates the GDOP contribution of each anchor node, which requires the inversion of a $d×d$ matrix each time. The complexity of the NSA module is $\mathcal{O}(md^3)$. In the ToA-based localization, This module involves solving linear least squares problems, where the complexity of matrix multiplication is $\mathcal{O}(md^2)$ and the complexity of matrix inversion is $\mathcal{O}(1)$. This module has a complexity of $\mathcal{O}(md^2)$. Combining the above, the total complexity is $\mathcal{O}(max \{m  k^3,md^3\})$. 

\section{PERFORMANCE EVALUTION}
\subsection{Simulation Setup}

In the 3-dimensional simulations, we employ a sensing field of 1000 m $\times$ 1000 m $\times$ 100 m with coordinates aligned to the Cartesian system (X, Y, Z axes corresponding to the dimensions of 1000 m, 1000 m and 100 m, respectively). 50 targets are randomly deployed in the field. The number of anchors varies from 6 to 30 to assess computation time and communication overhead. All nodes, including anchors and targets, can move or remain stationary. In addition, realistic positioning conditions are simulated by introducing zero-mean Gaussian noise into ToA measurements, the standard deviation of which is 6.1 ns \cite{35}. All experiments were performed on an Inter Xeon Silver 4210R platform, with critical system parameters listed in Table \ref{tableI}. A typical simulation scenario consisting of 6 anchors can be shown in Fig. \ref{fig_typical}, in which the target and the six anchor points are randomly distributed. Each anchor moves at a certain speed in a randomly chosen direction (assuming that no collisions occur among them) and the target remains still.

\begin{table}[!t]
\renewcommand{\arraystretch}{1.1}
\caption{Simulation Parameters}
\label{tableI}
\centering
\begin{tabular}{c||c}
\toprule
\textbf{Parameter} & \textbf{Value} \\
\hline
Localization field (m) & $1000 \times 1000 \times 100$ \\ \hline
Number of anchors & 6$\sim$30 \\ \hline
Number of targets & 1 \\ \hline
Simulation duration (s)& 10 \\ \hline
Standard deviation of ToA noise (ns) & 6.1 \\ \hline
Ciphertext representation (bit) & 1024 \\ \hline
Plaintext representation (bit) & 24 \\ \hline
Paillier modulus (bit) & 512 \\ \hline
Speed of nodes (m/s) & 0$\sim$10 \\
\bottomrule
\end{tabular}
\end{table}

\begin{figure}
\centering
\includegraphics[width=3in]{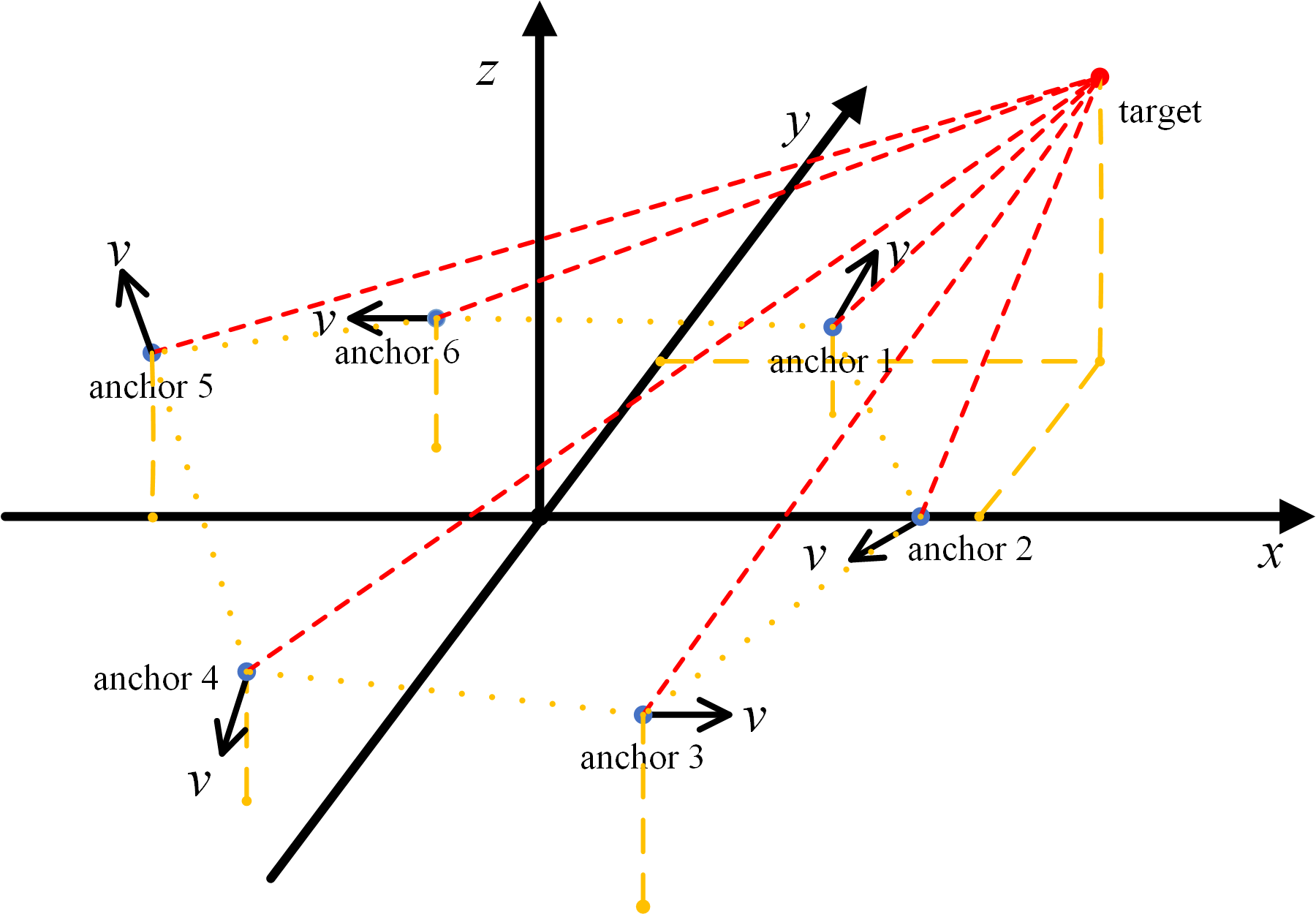}
\caption{A typical simulation scenario consisting of 6 anchors.}
\label{fig_typical}
\end{figure}

Our experiment adopts the optimized Paillier cryptosystem implementation in \cite{24}, utilizing cryptographic primitives including key generation via the \textit{paillier.generate}\_\textit{paillier}\_\textit{keypair} function, location data encryption/decryption through the \textit{public\_key.raw\_encrypt} function and \textit{private\_key.decrypt} function, sliding-window Montgomery modular exponentiation for index operations, and Toom-Cook-3 accelerated homomorphic operations (\textit{raw\_add/raw\_multiply}) \cite{36},\cite{37}. This method ensures consistency with security guarantees and computational efficiency benchmarks in all cryptographic phases.

\subsection{Algorithm Performance Under different selected anchor-number}
Before conducting comparative analysis with other schemes, we first determine the selected anchor-number $n$ to evaluate the effectiveness of NSA and examine its influence on algorithm performance. Experiments were performed using PPLZN with $n = 10, 15, 20, 25$, measuring computation time, communication overhead, and localization accuracy against the non-selective PPLZN baseline.

\textbf{Computation Overhead:} The total computation time as a function of anchor-number is shown in Fig. \ref{fig4}. Note that NSA is inactive when the number of anchors is below $n$. Thus, NSA introduces additional computational overhead beyond a certain anchor threshold, leading to a sharp increase in total computation time. However, compared to non-selective PPLZN, the computational cost of selective PPLZN increases more slowly. Moreover, a smaller $n$ results in lower computation time under large anchor counts, demonstrating that NSA effectively reduces computational overhead in such scenarios.

\begin{figure}[!t]
\centering
\includegraphics[width=3.5in]{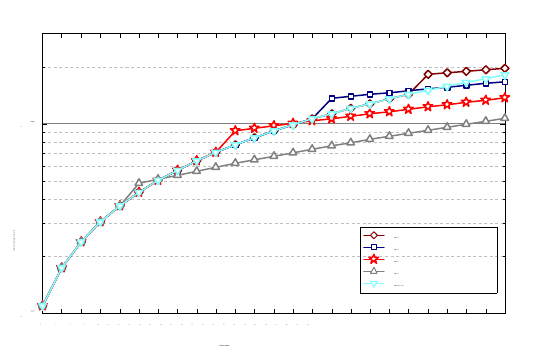}
\caption{Computation time of PPLZN under different parameter settings and the non-selective PPLZN.}
\label{chotimes} \label{fig4}
\end{figure}
\textbf{Communication Overhead:} As shown in Fig. \ref{fig5}, the communication overhead—measured in transmitted bits—is strongly influenced by the complexity of homomorphic encryption during localization. Without anchor selection, the number of ciphertext operations grows exponentially with the anchor count, leading to significantly higher communication costs. In contrast, the parameter $n$ limits the number of anchors used in localization, thereby constraining the number of ciphertext bits transmitted. Hence, the communication overhead increases only slowly with more number of anchors. Moreover, smaller values of $n$ yield better communication efficiency across different selective schemes.
\begin{figure}[!t]
\centering
\includegraphics[width=3.5in]{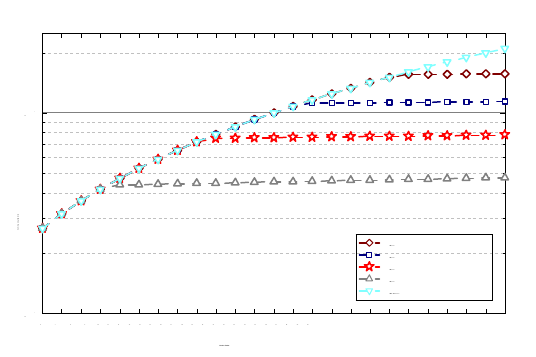}
\caption{Communication cost of PPLZN under different parameter settings and the non-selective PPLZN.}
\label{chobits} \label{fig5}
\end{figure}

\textbf{Location Accuracy:} Fig. \ref{fig6} compares the cumulative distribution functions of positioning error in a scenario with 30 anchors. Compared to raw ToA, all selective schemes exhibit marginally less accurate results, with improved precision as $n$ increases. Given the uniform observation accuracy across anchors, excluding any anchor via NSA inevitably leads to loss of positional information in the absence of prior knowledge.
\begin{figure}[!t]
\centering
\includegraphics[width=3.5in]{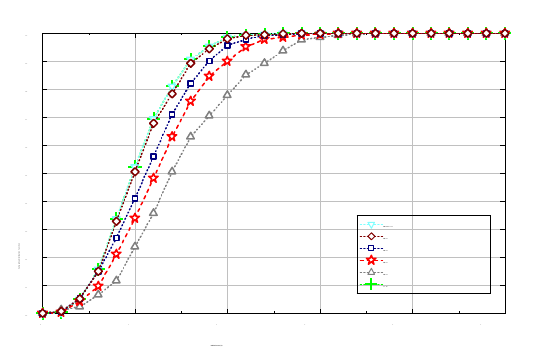}
\caption{Localization accuracy of PPLZN under different parameter settings and the non-selective PPLZN.}
\label{choerrors}  \label{fig6}
\end{figure}

Our objective is to maintain positioning accuracy within acceptable bounds while achieving high efficiency in both computation and communication. Theoretically, our scheme leverages zero-sum noise, which affords superior positioning accuracy compared to most alternative methods. This advantage allows a marginal sacrifice in localization precision in exchange for significantly improved computational and communication performance. Based on this three-way trade-off, we select $n = 15$ for subsequent comparative evaluation against baseline schemes.

\subsection{Numerical Results}
To demonstrate the advantages of our scheme, we compare it with three state-of-the-art privacy-preserving methods (EPPL \cite{39},  $\text P^3$-Pro \cite{24}, PPRP \cite{38}) as well as the conventional FHE approach \cite{40} on three key performance metrics.

\begin{figure*}[!t]
\centering
\subfloat[]{\includegraphics[width=2.4in]{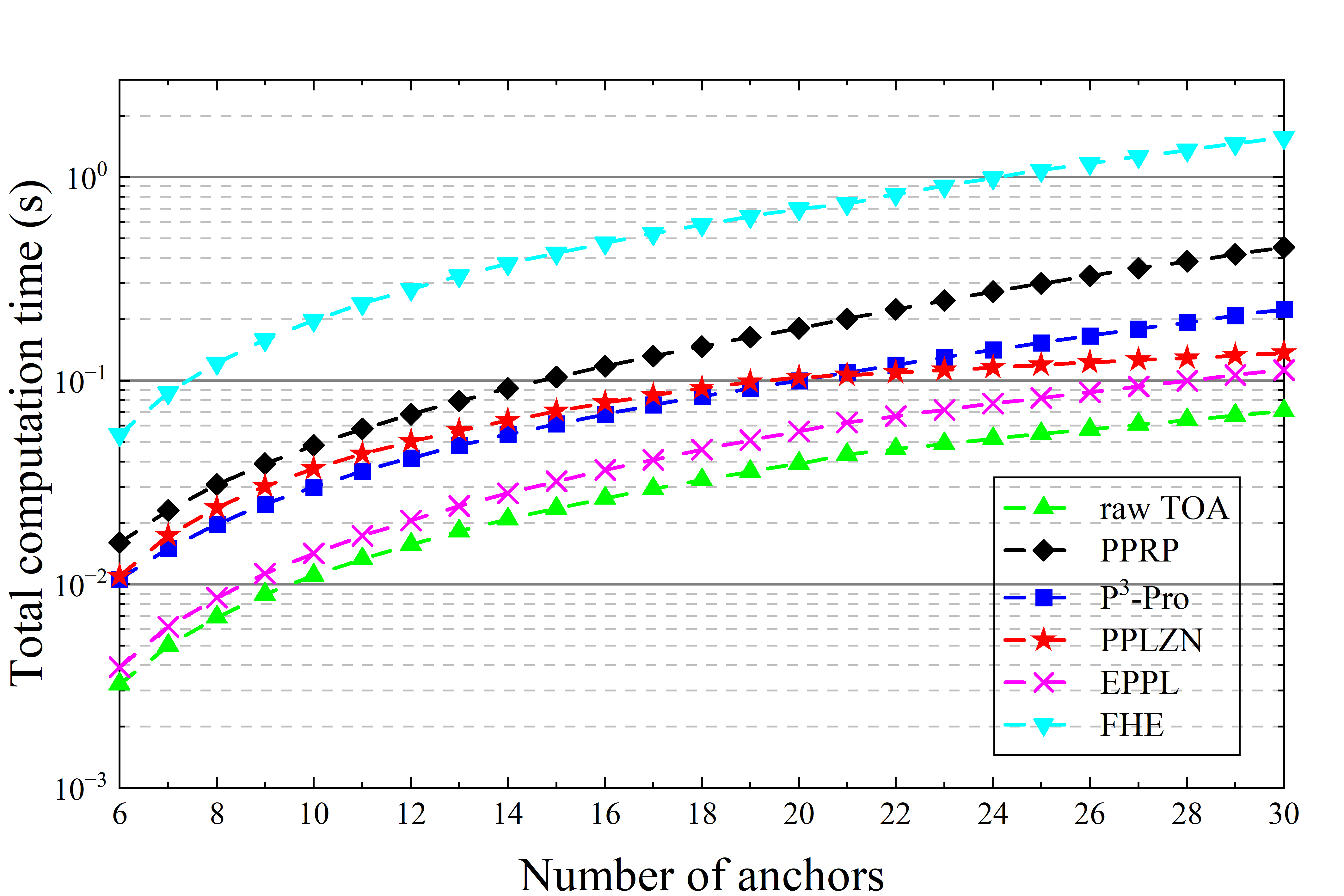}%
\label{fig_first_case}}
\subfloat[]{\includegraphics[width=2.4in]{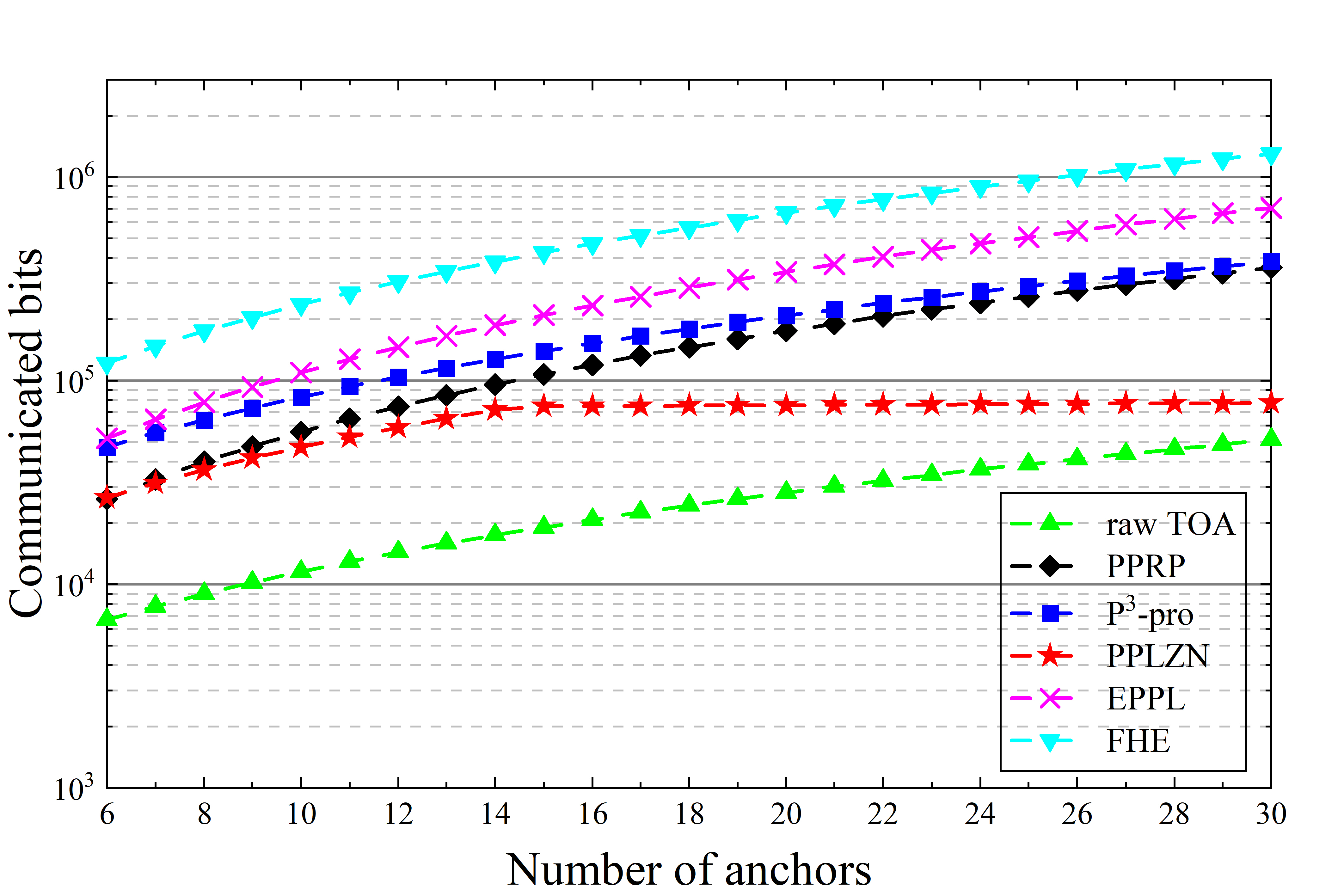}%
\label{fig_second_case}}
\subfloat[]{\includegraphics[width=2.4in]{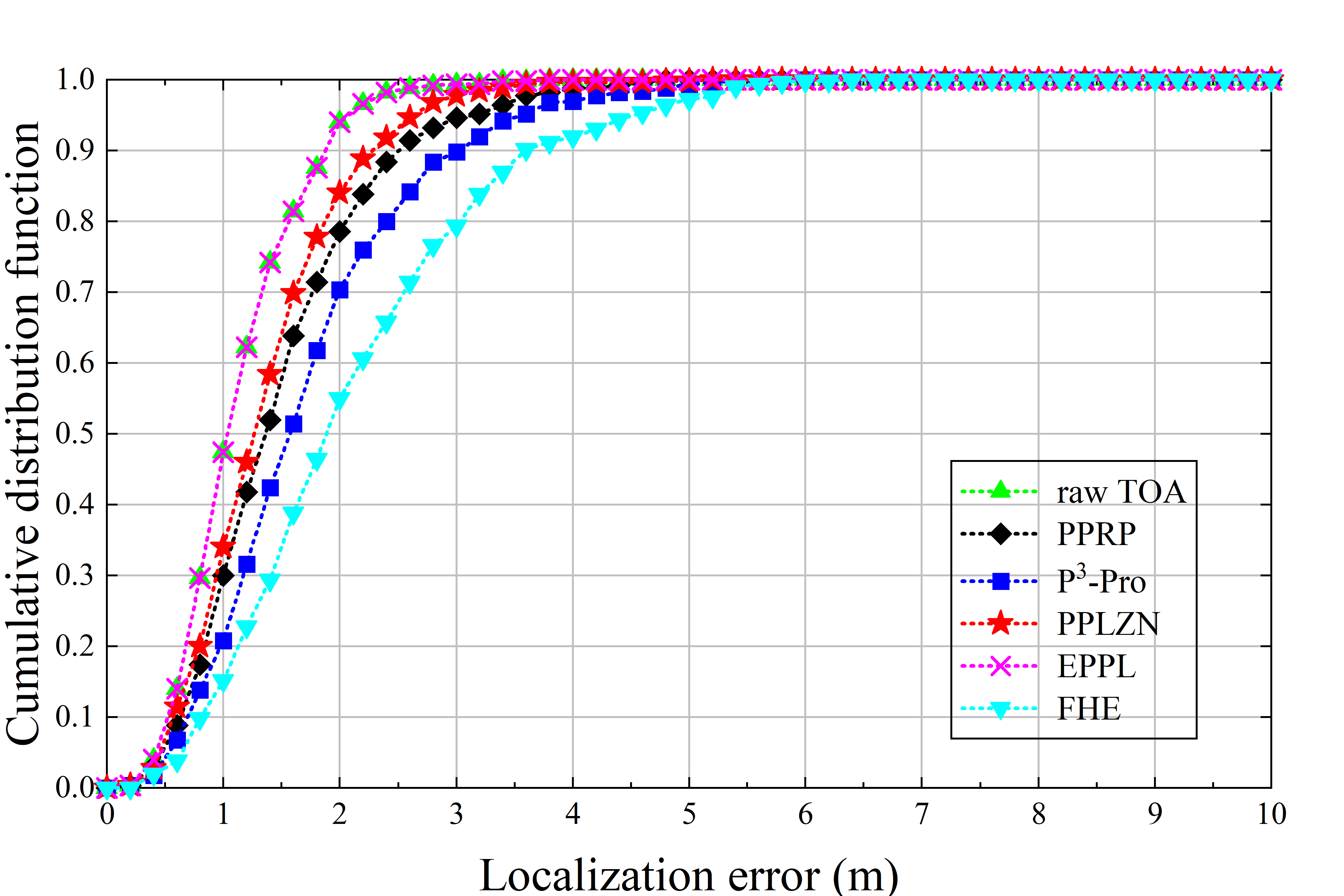}%
\label{fig_third_case}}
\caption{Comparison between different schemes. (a) Total computation time of different schemes. (b) Communication cost of different schemes. (c) Localization accuracy of different schemes.}
\label{fig_sim}
\end{figure*}
\textbf{Computation Overhead:} We first compare the time complexity of different schemes. As shown in Fig. \ref{fig_first_case}, the total computation time varies with the number of anchors. In the PPRP scheme, each anchor uploads location-related data and distance measurements to two location management function (LMF) servers using lightweight additive secret sharing (ASS), avoiding heavy cryptographic operations. The LMF servers then perform homomorphic matrix computations collaboratively to obtain targer's position without reconstructing raw data. However, as the anchor count rises, the number of non-linear operations grows exponentially due to the separate encryption and processing of two secret shares, leading to a rapid increase in computation time. The $\text P^3$-pro scheme primarily employs Shamir secret sharing (SSS) to obscure anchor locations and uses homomorphic encryption for server-side positioning \cite{24}. Since only a small subset of secrets require Paillier encryption, its overall computation time remains relatively low.In contrast, EPPL uses an adjacent subtraction-based model with matrix decomposition and zero-sum noise to achieve privacy without homomorphic encryption, resulting in optimal computational efficiency.The proposed PPLZN approach relies partially on Paillier encryption for generating zero-sum noise and encrypting sensitive data, while other steps use simpler zero-sum noise operations. Thus, its computational cost is dominated by homomorphic computations. When the number of anchors $m<20$, PPLZN performs slightly worse than $\text P^3$-pro; when $m\ge 20$, it outperforms $\text P^3$-pro, reducing total computation time by 45.5\% at $m = 30$. These results demonstrate that PPLZN significantly improves computational performance under high anchor counts.

\textbf{Communication Overhead:} As shown in Fig. \ref{fig_second_case}, the communication overhead—measured in transmitted bits—is directly influenced by cryptographic complexity. In PPRP, the secret from each anchor is split into two ciphertexts, and subsequent ciphertext operations are performed. As the number of anchors increases, these ciphertexts expand steadily, leading to moderate communication overhead. By contrast, $\text P^3$-Pro incurs higher communication overhead due to SSS, which requires each anchor to distribute shares to all others, resulting in increased data transmission. However, as the number of anchors grows, PPRP’s use of homomorphic encryption causes ciphertext expansion, so its communication cost gradually approaches that of $\text P^3$-Pro. EPPL exhibits the second highest communication overhead. Its broadcast-based zero-sum matrix distribution requires each anchor to transmit data to all others, followed by an aggregation step. This two-phase process significantly increases communication consumption, particularly with large numbers of anchors. PPLZN replaces private shares with zero-sum noise, reducing anchor-to-anchor transmissions by 26\% at $m=15$ compared to $\text P^3$-Pro. when $m\ge 15$, the communication overhead of PPLZN stabilizes near the baseline. This is because anchor selection intrinsically limits communication traffic scale, introducing only minimal additional overhead as the number of anchors increases.

\textbf{Location Accuracy:} As a primary goal of the localization, it is necessary to evaluate the performance of each privacy-preserving scheme under a fixed anchor count of 30. The cumulative distribution functions (CDFs) of the estimation errors are shown in Fig. \ref{fig_third_case}. The key advantage of zero-sum noise is its ability to preserve localization accuracy without cryptographic distortion. As Fig. \ref{fig_third_case} indicates, EPPL—which uses only zero-sum noise—achieves the same precision as raw ToA, owing to the self-canceling property of the noise during aggregation. In contrast, PPLZN introduces an approximately 15\%increase in RMSE compared to raw ToA due to anchor selection. Nonetheless, it still outperforms other cryptographic schemes such as $\text P^3$-Pro and PPRP. The significant accuracy loss in $\text P^3$-Pro stems from its Shamir Secret Sharing framework: reconstructing secrets through polynomial interpolation introduces approximation errors, especially with insufficient shares. Although PPRP uses theoretically lossless additive secret sharing, it suffers from quantization error during encryption and decryption. Converting floating-point coordinates to a finite integer domain truncates fractional values, leading to an average positioning drift of 0.35 m.

In summary, the proposed scheme PPLZN achieves strong communication efficiency and localization accuracy while maintaining competitive computational performance in practical settings.

\subsection{Privacy-Preserving Evaluation}
Based on the aforementioned privacy-preserving objectives, the privacy-preserving evaluation in this work is divided into three hierarchical levels:
\begin{itemize}
    \item \textbf{Anchor-to-Anchor}: Prevents any anchor $\mathbb{A}_i$ from accessing the location information of any other anchor $\mathbb{A}_j$ where $i\neq j$.
    \item \textbf{Target-to-Anchor}: Ensures mutual privacy where the target cannot obtain anchor $\mathbb{A}_i$'s location, and no anchor $\mathbb{A}_j$ can obtain the target's position.
    \item \textbf{Node-to-Third Party}: Ensures that any third-party server or aggregator processing positioning data cannot deduce the locations of either targets or anchors.
\end{itemize}

\begin{table}[!t]
\renewcommand{\arraystretch}{1.1}
\centering
\caption{Performance Summary Of Different Privacy-Preserving schemes} 
\label{tableII} 
\begin{tabular}{|c|c|c|c|} 
\hline
\multirow{2}*{scheme} & \multicolumn{3}{c|}{Privacy Goal} \\ 
\cline{2-4} 
& Anchor-to-Anchor & Target-to-Anchor & Node-to-Third Party \\
\hline
PPLZN & \checkmark & \checkmark & \checkmark \\
\hline
PPRP & \checkmark & \checkmark & $\times$ \\ 
\hline
$\text P^3$-pro & \checkmark & \checkmark & \checkmark \\
\hline
EPPL & \checkmark & \checkmark & N/A \\ 
\hline
FHE & \checkmark & \checkmark & \checkmark \\
\hline
\end{tabular}
\end{table}

Table \ref{tableII} provides a comprehensive comparison of the privacy-preserving capabilities of our scheme alongside four benchmark methods. All schemes satisfy the primary requirements for anchor-to-anchor and target-to-anchor privacy protection. In PPLZN, aggregators process data perturbed by zero-sum noise, thereby preventing location disclosure, while Paillier homomorphic encryption ensures computational confidentiality. $\text P^3$-Pro similarly combines SSS and Paillier cryptosystems to preserve location privacy. Although PPRP distributes location data across two servers to reduce the risks of single-point failures or malicious attacks, the system remains vulnerable to information compromise if both servers are breached. EPPL operates without third-party involvement, making this category not applicable (N/A). The FHE scheme, implemented via the Gentry algorithm \cite{40}, utilizes fully homomorphic encryption that supports both addition and multiplication, thereby achieving all three levels of privacy protection.

\section{CONCLUSION}
This study enhances collaborative localization performance through PPLZN, a novel privacy-preserving scheme that integrates zero-sum noise with Paillier Homomorphic Encryption. By ensuring mutual position confidentiality among all participating entities under the honest-but-curious model, PPLZN achieves robust privacy protection while maintaining high positioning accuracy. Key innovations include a cryptographic zero-sum noise mechanism that masks sensitive data yet allows noise cancellation during position estimation, along with the NSA that dynamically optimizes anchor selection to sustain efficiency in dense networks, such as UAV networks. The performance analysis demonstrates significant advantages over existing schemes. Specifically, when the number of anchors reaches 30, PPLZN reduces computational overhead by more than 45.5\% compared to PPRP. At 15 anchors, it reduces communication traffic by over 26\% compared to $\text P^3$-Pro. Although the RMSE increases by approximately 15\% relative to raw ToA, PPLZN still achieves superior positioning accuracy compared to other privacy-preserving schemes. Overall, this work presents an efficient and secure solution for collaborative localization in privacy-sensitive environments.

\bibliographystyle{IEEEtran}
\bibliography{sample.bib}

\begin{thebibliography}{10}
\providecommand{\url}[1]{#1}
\csname url@samestyle\endcsname
\providecommand{\newblock}{\relax}
\providecommand{\bibinfo}[2]{#2}
\providecommand{\BIBentrySTDinterwordspacing}{\spaceskip=0pt\relax}
\providecommand{\BIBentryALTinterwordstretchfactor}{4}
\providecommand{\BIBentryALTinterwordspacing}{\spaceskip=\fontdimen2\font plus
\BIBentryALTinterwordstretchfactor\fontdimen3\font minus \fontdimen4\font\relax}
\providecommand{\BIBforeignlanguage}[2]{{%
\expandafter\ifx\csname l@#1\endcsname\relax
\typeout{** WARNING: IEEEtran.bst: No hyphenation pattern has been}%
\typeout{** loaded for the language `#1'. Using the pattern for}%
\typeout{** the default language instead.}%
\else
\language=\csname l@#1\endcsname
\fi
#2}}
\providecommand{\BIBdecl}{\relax}
\BIBdecl

\bibitem{1}
T.~Wang, Y.~Tao, Q.~Zhang, N.~Xu, F.~Chen, and C.~Zhao, ``Group coding location privacy protection method based on differential privacy in crowdsensing,'' \emph{IEEE Internet of Things Journal}, vol.~11, no.~17, pp. 28\,398--28\,408, 2024.

\bibitem{2}
F.~Alam, N.~Faulkner, and B.~Parr, ``Device-free localization: A review of non-rf techniques for unobtrusive indoor positioning,'' \emph{IEEE Internet of Things Journal}, vol.~8, no.~6, pp. 4228--4249, 2020.

\bibitem{3}
C.~Xiang, S.~Zhang, S.~Xu, and G.~Mao, ``Crowdsourcing-based indoor localization with knowledge-aided fingerprint transfer,'' \emph{IEEE Sensors Journal}, vol.~22, no.~5, pp. 4281--4293, 2022.

\bibitem{5}
L.~L.~d. Oliveira, G.~H. Eisenkraemer, E.~A. Carara, J.~B. Martins, and J.~Monteiro, ``Mobile localization techniques for wireless sensor networks: Survey and recommendations,'' \emph{ACM Transactions on Sensor Networks}, vol.~19, no.~2, pp. 1--39, 2023.

\bibitem{13}
F.~Zuo, Y.~Li, G.~Wang, and X.~He, ``Towards accurate and privacy-preserving localization using anchor quality assessment in internet of things,'' \emph{Future Generation Computer Systems}, vol. 148, pp. 524--537, 2023.

\bibitem{39}
G.~Wang, J.~He, X.~Shi, J.~Pan, and S.~Shen, ``Analyzing and evaluating efficient privacy-preserving localization for pervasive computing,'' \emph{IEEE Internet of Things Journal}, vol.~5, no.~4, pp. 2993--3007, 2017.

\bibitem{1103}
J.~Yan, Y.~Meng, X.~Yang, X.~Luo, and X.~Guan, ``Privacy-preserving localization for underwater sensor networks via deep reinforcement learning,'' \emph{IEEE Transactions on information forensics and security}, vol.~16, pp. 1880--1895, 2020.

\bibitem{1-3}
H.~Jiang, J.~Li, P.~Zhao, F.~Zeng, Z.~Xiao, and A.~Iyengar, ``Location privacy-preserving mechanisms in location-based services: A comprehensive survey,'' \emph{ACM Computing Surveys (CSUR)}, vol.~54, no.~1, pp. 1--36, 2021.

\bibitem{1-4}
M.~Atif, R.~Ahmad, W.~Ahmad, L.~Zhao, and J.~J. Rodrigues, ``Uav-assisted wireless localization for search and rescue,'' \emph{IEEE Systems Journal}, vol.~15, no.~3, pp. 3261--3272, 2021.

\bibitem{1-5}
F.~Yessoufou, E.~Chicha, S.~Sassi, R.~Chbeir, and J.~Hounsou, ``Crowdpred: Privacy-preserving approach for locations on decentralized crowdsourcing application,'' in \emph{2023 International Conference on Innovations in Intelligent Systems and Applications (INISTA)}.\hskip 1em plus 0.5em minus 0.4em\relax IEEE, 2023, pp. 1--6.

\bibitem{7}
S.~Halder and T.~Newe, ``Enabling secure time-series data sharing via homomorphic encryption in cloud-assisted iiot,'' \emph{Future Generation Computer Systems}, vol. 133, pp. 351--363, 2022.

\bibitem{8}
B.~Zeng, X.~Yan, X.~Zhang, and B.~Zhao, ``Brake: Bilateral privacy-preserving and accurate task assignment in fog-assisted mobile crowdsensing,'' \emph{IEEE Systems Journal}, vol.~15, no.~3, pp. 4480--4491, 2020.

\bibitem{9}
Y.~Zhu and J.~Hu, ``To hide anchor’s position in range-based wireless localization via secret sharing,'' \emph{IEEE Wireless Communications Letters}, vol.~11, no.~7, pp. 1325--1328, 2022.

\bibitem{10}
W.~Wang, Y.~Wang, P.~Duan, T.~Liu, X.~Tong, and Z.~Cai, ``A triple real-time trajectory privacy protection mechanism based on edge computing and blockchain in mobile crowdsourcing,'' \emph{IEEE Transactions on Mobile Computing}, vol.~22, no.~10, pp. 5625--5642, 2022.

\bibitem{11}
Y.~Wang, M.~Huang, Q.~Jin, and J.~Ma, ``Dp3: A differential privacy-based privacy-preserving indoor localization mechanism,'' \emph{IEEE Communications Letters}, vol.~22, no.~12, pp. 2547--2550, 2018.

\bibitem{24}
Y.~Zhu, Y.~Qiu, J.~Wang, J.~Hu, F.~Yan, and S.~Zhao, ``Protecting position privacy in range-based crowdsourcing cooperative localization,'' \emph{IEEE Transactions on Network Science and Engineering}, vol.~11, no.~1, pp. 1136--1150, 2023.

\bibitem{12}
Y.~Li, G.~Wang, and F.~Zuo, ``Efficient privacy preserving single anchor localization using noise-adding mechanism for internet of things,'' in \emph{International Conference on Web Information Systems and Applications}.\hskip 1em plus 0.5em minus 0.4em\relax Springer, 2021, pp. 261--273.

\bibitem{14}
G.~Wang, X.~Zhang, and Y.~Li, ``Design and analysis of privacy-preserving localization assisted by reconfigurable intelligent surface for internet of things,'' in \emph{Proceedings of the 2023 11th International Conference on Communications and Broadband Networking}, 2023, pp. 1--7.

\bibitem{16}
S.~Li, H.~Li, and L.~Sun, ``Privacy-preserving crowdsourced site survey in wifi fingerprint-based localization,'' \emph{EURASIP Journal on Wireless Communications and Networking}, vol. 2016, no.~1, p. 123, 2016.

\bibitem{15}
N.~Alikhani, V.~Moghtadaiee, A.~M. Sazdar, and S.~A. Ghorashi, ``A privacy preserving method for crowdsourcing in indoor fingerprinting localization,'' in \emph{2018 8th International Conference on Computer and Knowledge Engineering (ICCKE)}.\hskip 1em plus 0.5em minus 0.4em\relax IEEE, 2018, pp. 58--62.

\bibitem{17}
R.~Nieminen and K.~J{\"a}rvinen, ``Practical privacy-preserving indoor localization based on secure two-party computation,'' \emph{IEEE Transactions on Mobile Computing}, vol.~20, no.~9, pp. 2877--2890, 2020.

\bibitem{18}
G.~Wang, Y.~Li, R.~Liu, F.~Tong, J.~Pan, F.~Zuo, and X.~He, ``Enhancing privacy-preserving localization by integrating random noise with blockchain in internet of things,'' \emph{IEEE Transactions on Network and Service Management}, vol.~21, no.~2, pp. 2445--2459, 2023.

\bibitem{1021}
L.~Xie, L.~Li, X.~Yang, J.~Wang, and M.~Zhou, ``A low communication overhead privacy preserving collaboration localization method for asynchronous networks,'' \emph{IEEE Transactions on Cognitive Communications and Networking}, 2025.

\bibitem{19}
S.~Xu, L.~Wu, K.~Do{\u{g}}an{\c{c}}ay, and M.~Alaee-Kerahroodi, ``A hybrid approach to optimal toa-sensor placement with fixed shared sensors for simultaneous multi-target localization,'' \emph{IEEE Transactions on Signal Processing}, vol.~70, pp. 1197--1212, 2022.

\bibitem{20}
P.~Paillier, ``Public-key cryptosystems based on composite degree residuosity classes,'' in \emph{International conference on the theory and applications of cryptographic techniques}.\hskip 1em plus 0.5em minus 0.4em\relax Springer, 1999, pp. 223--238.

\bibitem{21}
R.~Ben~Romdhane, H.~Hammami, M.~Hamdi, and T.-H. Kim, ``Privacy-preserving spatial and temporal data aggregation for smart metering,'' in \emph{Proceedings of the Asia conference on electrical, power and computer engineering}, 2022, pp. 1--4.

\bibitem{22}
H.~Li, L.~Sun, H.~Zhu, X.~Lu, and X.~Cheng, ``Achieving privacy preservation in wifi fingerprint-based localization,'' in \emph{Ieee Infocom 2014-IEEE Conference on Computer Communications}.\hskip 1em plus 0.5em minus 0.4em\relax IEEE, 2014, pp. 2337--2345.

\bibitem{23}
S.~Zhang, T.~Zheng, and B.~Wang, ``A privacy protection scheme for smart meter that can verify terminal’s trustworthiness,'' \emph{International Journal of Electrical Power \& Energy Systems}, vol. 108, pp. 117--124, 2019.

\bibitem{25}
W.~Li, M.~Liu, T.~Chen, and G.~Mao, ``Vehicles selection algorithm for cooperative localization based on stochastic geometry in internet of vehicle systems,'' \emph{IEEE Transactions on Vehicular Technology}, 2024.

\bibitem{26}
I.~Sharp, K.~Yu, and Y.~J. Guo, ``Gdop analysis for positioning system design,'' \emph{IEEE Transactions on Vehicular Technology}, vol.~58, no.~7, pp. 3371--3382, 2009.

\bibitem{27}
X.~Lv, K.~Liu, and P.~Hu, ``Geometry influence on gdop in toa and aoa positioning systems,'' in \emph{2010 Second International Conference on Networks Security, Wireless Communications and Trusted Computing}, vol.~2.\hskip 1em plus 0.5em minus 0.4em\relax IEEE, 2010, pp. 58--61.

\bibitem{29}
T.~Shu, Y.~Chen, and J.~Yang, ``Protecting multi-lateral localization privacy in pervasive environments,'' \emph{IEEE/ACM transactions on networking}, vol.~23, no.~5, pp. 1688--1701, 2015.

\bibitem{34}
G.~Wang, J.~He, X.~Shi, J.~Pan, and S.~Shen, ``Analyzing and evaluating efficient privacy-preserving localization for pervasive computing,'' \emph{IEEE Internet of Things Journal}, vol.~5, no.~4, pp. 2993--3007, 2017.

\bibitem{CRLB_TOA}
T.~Jia and R.~M. Buehrer, ``A new cramer-rao lower bound for toa-based localization,'' in \emph{MILCOM 2008-2008 IEEE military communications conference}.\hskip 1em plus 0.5em minus 0.4em\relax IEEE, 2008, pp. 1--5.

\bibitem{28}
J.~Shi, K.~Li, L.~Chai, L.~Liang, C.~Tian, and K.~Xu, ``Fast satellite selection algorithm for gnss multi-system based on sherman--morrison formula,'' \emph{GPS Solutions}, vol.~27, no.~1, p.~44, 2023.

\bibitem{36}
P.~Paillier, ``Public-key cryptosystems based on composite degree residuosity classes,'' in \emph{International conference on the theory and applications of cryptographic techniques}.\hskip 1em plus 0.5em minus 0.4em\relax Springer, 1999, pp. 223--238.

\bibitem{35}
N.~Patwari, A.~O. Hero, M.~Perkins, N.~S. Correal, and R.~J. O'dea, ``Relative location estimation in wireless sensor networks,'' \emph{IEEE Transactions on signal processing}, vol.~51, no.~8, pp. 2137--2148, 2003.

\bibitem{37}
I.~Damgard, M.~Jurik, and J.~Nielsen, ``A generalization of paillier’s public-key system with applications to electronic voting, 2003,'' \emph{Int. J. Inf. Secur}, vol.~9, no.~6, pp. 371--385, 2010.

\bibitem{38}
C.~Huang, D.~Liu, A.~Yang, R.~Lu, and X.~Shen, ``Pprp: preserving location privacy for range-based positioning in mobile networks,'' \emph{IEEE Transactions on Mobile Computing}, vol.~23, no.~10, pp. 9451--9468, 2024.

\bibitem{40}
Y.~Zhang, R.~Liu, and D.~Lin, ``Improved key generation algorithm for gentry’s fully homomorphic encryption scheme,'' in \emph{International Conference on Information Security and Cryptology}.\hskip 1em plus 0.5em minus 0.4em\relax Springer, 2017, pp. 93--111.

\end{thebibliography}

\vfill

\end{document}


{\appendix[Proof of Proposition 2]  
\label{app:A}

To prove Proposition 2 is to prove $\mathbf x = (\mathbf A^\text{T}\mathbf A)^{-1} (\mathbf c + v^{2} \cdot {Decr}(\mathbf e)). $ According to the zero-sum noise mechanism, $\mathbf A^\text{T}\mathbf A$ and $\mathbf c$  are apparently equal before and after encryption. In the Paillier encryption phase, $ \mathbf{e} $ is encrypted by a public key before being sent to the aggregator, and is eventually decrypted on the target side. Note that, in this process, we just need to prove $ \mathbf A^\text{T}\mathbf{b} = \mathbf c + v^{2} \cdot {Decr}(\mathbf e) $.

Since $ \mathbf A^\text{T}\mathbf{b}$ is a $4 \times 1$ matrix, $ \mathbf{b}$ can be calculated by (40). Then, let $\alpha_{ij}$ be the $j$th term in $\alpha_{i}$, $t_{0 i} \oplus t_{i i}$ is denoted as $t_{i}$. Based on (44), it can be known that 
\begin{align}
T_{0i} \otimes \chi_i
&= T_{0i} \otimes\alpha_i \otimes (t_{0 i} \oplus t_{i i})\nonumber\\
&=
\begin{bmatrix}
 \alpha_{i1}T_{0i} \otimes t_{i}\\
 \alpha_{i2}T_{0i} \otimes t_{i} \\
 \alpha_{i3}T_{0i} \otimes t_{i} \\ 
 \alpha_{i4}T_{0i} \otimes t_{i}
\end{bmatrix} =
\begin{bmatrix}
t_{i}^{\alpha_{i1}T_{0i}}\\
t_{i}^{\alpha_{i2}T_{0i}}\\
t_{i}^{\alpha_{i3}T_{0i}}\\
t_{i}^{\alpha_{i4}T_{0i}}
\end{bmatrix},
\end{align}
$\otimes$ and $\oplus$ represent homomorphic multiplication and addition respectively, and then $ \mathbf e $ can be represented as
\begin{align}
\mathbf{e} 
&= (T_{01} \otimes \chi_1) \oplus (T_{02} \otimes \chi_2) \oplus \cdots \oplus (T_{0 m} \otimes \chi_m) \nonumber \\
&=
\begin{bmatrix} 
t_{1}^{\alpha_{i1}T_{01}}\oplus \cdots \oplus t_{m}^{\alpha_{i1}T_{0m}}\\
t_{1}^{\alpha_{i2}T_{01}}\oplus \cdots \oplus t_{m}^{\alpha_{i2}T_{0m}}\\
t_{1}^{\alpha_{i3}T_{01}}\oplus \cdots \oplus t_{m}^{\alpha_{i3}T_{0m}}\\
t_{1}^{\alpha_{i4}T_{01}}\oplus \cdots \oplus t_{m}^{\alpha_{i4}T_{0m}}
\end{bmatrix} 
= 
\begin{bmatrix} 
\prod_{i=1}^m t_{i}^{\alpha_{i1}T_{0i}}\\
\prod_{i=1}^m t_{i}^{\alpha_{i2}T_{0i}}\\
\prod_{i=1}^m t_{i}^{\alpha_{i3}T_{0i}}\\
\prod_{i=1}^m t_{i}^{\alpha_{i4}T_{0i}}
\end{bmatrix},
\end{align}
and $ T_{0 i} - 2T_{i} $ is encrypted by a public-key $ (n,g) $ as
\begin{equation}
\llbracket (T_{0 i} - 2T_{i}) \rrbracket_{pk} = t_{i}=g^{T_{0 i} - 2T_{i}}r^n \bmod n^2.
\end{equation}

\begin{lem}
If $ n = pq $ with $ p $ and $ q $ are two big primes, then for any $ y \in \mathbb{Z}_{n^{2}}^{*} $, it has the following properties:
\begin{equation}
\begin{cases}
y^{\lambda(n)} = 1 \bmod{n} \\
y^{n\lambda(n)} = 1 \bmod{n^{2}},
\end{cases}
\end{equation}
where $ \lambda(n) = \text{lcm}(p - 1, q - 1) $ is the Carmichael function.
\end{lem}

\begin{proof}
Since $ y $ and $ n $ are coprime, according to Euler's theorem, we have
\begin{equation}
y^{\lambda(n)} = 1 \bmod{n}.
\end{equation}
Then, according to Carmichael's theorem, we have
\begin{align}
\lambda(n^{2}) &= \text{lcm}(\lambda(p^{2}), \lambda(q^{2})) 
= \text{lcm}(\phi(p^{2}), \phi(q^{2})) \nonumber \\
&= \text{lcm}(p(p - 1), q(q - 1)) \nonumber \\
&= pq \text{lcm}(p - 1, q - 1) \nonumber \\
&= n\lambda(n),
\end{align}
where $ \phi(n) $ is Euler's totient function, representing the number of positive integers in $ \mathbb{Z}_{n}^{*} $ that are coprime to $ n $.

Therefore,
\begin{equation}
y^{n\lambda(n)} = y^{\lambda(n^{2})} = 1 \bmod{n^{2}}.
\end{equation}
\end{proof}
For the $j$th term in $ \mathbf e $ (denoted as $ \mathbf e_j(j=1,2,3,4) $ ), it can be calculated by the aggregator as
\begin{align}
\mathbf{e}_j&= \prod_{i=1}^m t_{i}^{\alpha_{ij}T_{0i}} \nonumber \\
&= \prod_{i=1}^m (g^{T_{0 i} - 2T_{i}}r^{n})^{\alpha_{ij}T_{0i}} \nonumber \\
&= g^{\sum_{i=1}^m \alpha_{ij}T_{0 i} (T_{0 i} - 2 T_i)} r^{n\sum_{i=1}^m \alpha_{ij}T_{0 i}} \bmod{n^{2}}.
\end{align}
Then,
\begin{align}
\mathbf{e}^{\lambda}_j &= \left(g^{\sum_{i=1}^m \alpha_{ij}T_{0 i} (T_{0 i} - 2 T_i)} r^{n\sum_{i=1}^m \alpha_{ij}T_{0 i}} \right)^{\lambda} \nonumber \\
&= g^{\lambda \sum_{i=1}^m \alpha_{ij}T_{0 i} (T_{0 i} - 2 T_i)} r^{n\lambda \sum_{i=1}^m \alpha_{ij}T_{0 i}} \nonumber \\
&= g^{\lambda \sum_{i=1}^m \alpha_{ij}T_{0 i} (T_{0 i} - 2 T_i)} \bmod{n^2} \quad \text{(by Lemma 1)} \nonumber \\
&= g^{\tilde{\mathcal{A}}_j \lambda} \bmod{n^2}.
\end{align}
Let $\tilde{\mathcal{A}}_j = \sum_{i=1}^m \alpha_{ij}T_{0 i} (T_{0 i} - 2 T_i)$ and we apply Taylor expansion of $ g^{\tilde{\mathcal{A}}_j \lambda} $:
\begin{align}
g^{\tilde{\mathcal{A}}_j \lambda} &= (1 + (g - 1))^{\tilde{\mathcal{A}}_j \lambda} \nonumber \\
&= \sum_{l=0}^{\tilde{\mathcal{A}}_j \lambda} \binom{\tilde{\mathcal{A}}_j \lambda}{l} (g - 1)^l \nonumber \\
&= 1 + (g - 1)\tilde{\mathcal{A}}_j \lambda + \binom{\tilde{\mathcal{A}}_j \lambda}{2} (g - 1)^2 + \cdots,
\end{align}
where $\binom{n}{k}$ is the binomial coefficient. Because $ g $ is selected from $ \mathbb{Z}_{n^{2}}^{*} $ and satisfies
\begin{equation}
\gcd\left(L\left(g^{\lambda} \bmod{n^{2}}\right), n\right) = 1,
\end{equation}
where 
\begin{equation}
L(x) = \frac{x - 1}{n}.  
\end{equation}
Taking a simple example, let \( g = n + 1 \), then
\begin{equation}
g^{\tilde{\mathcal{A}}_j \lambda} \bmod n^{2} = 1 + n\tilde{\mathcal{A}}_j \lambda \bmod n^{2}. 
\end{equation}
Similarly,
\begin{equation}
g^{\lambda} \bmod n^{2} = 1 + n\lambda \bmod n^{2}. 
\end{equation}

When the target receives the encrypted value \( [e_{1}, e_{2}, e_{3}, e_{4}]^\text{T}\) from the aggregator, it decrypts them using its secret-key \( (\lambda, \alpha) \) as
\begin{equation}
 {Decr}(\mathbf{e}_{j}) = L(\mathbf{e}_{j}^{\lambda} \bmod n^{2}) \cdot \alpha \bmod n,  
\end{equation}
where
\begin{equation}
\alpha = \frac{1}{L(g^{\lambda} \bmod n^{2})}.  
\end{equation}
The decrypted result is given as
\begin{equation}
{Decr}(\mathbf{e}_{j}) = \frac{L(\mathbf{e}_{j}^{\lambda} \bmod n^{2})}{L(g^{\lambda} \bmod n^{2})} \bmod n = \tilde{\mathcal{A}}_j.  
\end{equation}
Then, \\
\begin{equation}
{Decr}(\mathbf e) =
\begin{bmatrix} 
{Decr}(\mathbf e_1)\\
{Decr}(\mathbf e_2)\\
{Decr}(\mathbf e_3)\\
{Decr}(\mathbf e_4)
\end{bmatrix}
=
\begin{bmatrix} 
\tilde{\mathcal{A}}_1\\
\tilde{\mathcal{A}}_2\\
\tilde{\mathcal{A}}_3\\
\tilde{\mathcal{A}}_4
\end{bmatrix}.  \label{eq56}
\end{equation}
Take \myref{eq56} in (40):
\begin{equation}
\mathbf A^\text{T}\mathbf{b} = \mathbf c + v^{2} \cdot {Decr}(\mathbf e) =\mathbf c + v^{2} \cdot \sum_{i=1}^m T_{0 i} \alpha_i (T_{0 i} - 2 T_i).
\end{equation}
We can see that $ \mathbf A^\text{T}\mathbf{b}$ is exactly the same as $\mathbf c + v^{2} \cdot {Decr}(\mathbf e)$, and thus the proof of Proposition 2 is completed.

}

\vfill